\documentclass[showpacs,amsmath,amsfonts,amssymb,aps,superscriptaddress]{revtex4}

\usepackage{makecell, tabularx}

\usepackage{booktabs}
\usepackage[dvipsnames, x11names]{xcolor}

\usepackage{comment}
\usepackage{enumitem}
\usepackage{graphicx}
\usepackage{float}
\usepackage{dcolumn}
\usepackage{bm}
\usepackage{mathrsfs}
\usepackage{amsmath}
\newtheorem{theorem}{Theorem}[section]

\newenvironment{proof}[1][Proof.]{\begin{trivlist}
\item[\hskip \labelsep {\bfseries #1}]}{\end{trivlist}}

\newcommand{\qedsymbol}{\ensuremath{\Box}}
\usepackage{amsfonts}
\usepackage{dsfont}
\makeatletter
\DeclareFontShape{OT1}{cmr}{m}{scit}{<->ssub * cmr/m/sc}{}
\makeatother
\begin{document}

\title{A Spectral Approach for Quasinormal Frequencies of Noncommutative Geometry-inspired Wormholes}
\author{Davide Batic}
\email{(Corresponding-author) davide.batic@ku.ac.ae}
\affiliation{
Mathematics Department, Khalifa University of Science and Technology, PO Box 127788, Abu Dhabi, United Arab Emirates}
\author{Denys Dutykh}
\email{denys.dutykh@ku.ac.ae}
\affiliation{
Mathematics Department, Khalifa University of Science and Technology, PO Box 127788, Abu Dhabi, United Arab Emirates}

\author{Joudy Jamal Beek}
\email{100061683@ku.ac.ae}
\affiliation{
Physics Department, Khalifa University of Science and Technology, PO Box 127788, Abu Dhabi, United Arab Emirates}
\date{\today}

\begin{abstract}
We present a detailed investigation of quasinormal modes (QNMs) for noncommutative geometry-inspired wormholes, focusing on scalar, electromagnetic, and vector-type gravitational perturbations. By employing the spectral method, the perturbation equations are reformulated into an eigenvalue problem over a compact domain, using Chebyshev polynomials to ensure high precision and fast numerical convergence. Our results reveal the absence of overdamped modes, with all detected QNMs exhibiting oscillatory behaviour. Additionally, for large values of the rescaled mass parameter, the QNMs of the noncommutative wormhole transition smoothly to those of the classical Schwarzschild wormhole, validating the accuracy of the spectral method. This work represents the first comprehensive exploration of QNMs in noncommutative geometry-inspired wormholes, shedding light on their stability and dynamical properties.

---
\end{abstract}
\maketitle
\section{Introduction}

Wormholes are fascinating theoretical gravitational objects predicted by general relativity that represent hypothetical shortcuts connecting distinct regions of spacetime \cite{Flamm1916PZ, Scholz2001, Einstein1935PR}. The first traversable wormhole solutions were independently proposed by Ellis \cite{Ellis1973JMP, Ellis1974JMP} and Bronnikov \cite{Bronnikov1973APP}, involving exotic configurations that required hypothetical matter violating certain energy conditions. In 1988, Morris and Thorne \cite{Morris1988AJP} derived a static traversable wormhole, identified as a specific case of the Ellis-Bronnikov manifold.  Solutions in higher dimensions, such as those derived within Einstein-Gauss-Bonnet theory \cite{Bhawal1992PRD, Dotti2007PRD}, have expanded the scope of wormhole research. Other developments include wormholes localized on branes \cite{Anchordoqui2000PRD, Bronnikov2003PRD, LaCamera2003PLB, Lobo2007PRD}, solutions in Brans-Dicke theory \cite{Nandi1998PRD}, and configurations within semi-classical gravity frameworks (see Ref. \cite{Garattini2007CQG} and references therein). Using a systematic geometric approach, exact wormhole solutions were also identified \cite{Boehmer2007PRD, Boehmer2008CQG}. Additionally, models supported by equations of state linked to cosmic acceleration have been explored \cite{Sushkov2005PRD, Lobo2005PRD, Lobo2005PRDa, Lobo2006PRD, Lobo2007PRDa}, while geometries that respect the Null Energy Condition (NEC) have been further investigated in the context of conformal Weyl gravity \cite{Lobo2008CQG}. For a comprehensive discussion and recent developments in these areas, see \cite{Lemos2003PRD, Lobo2008CQGR}. More recently, \cite{Wall2013CQG} invoked the generalized second law of causal horizons to rule out traversable wormholes connecting two disjoint regions of spacetime. In 2017, \cite{Gao2017JHEP} constructed short-lived, non-exotic traversable wormholes within the AdS/CFT correspondence, while \cite{Fu2019CQG} developed a perturbative wormhole solution beyond the constraints of this duality. Finally, traversable wormholes in Einstein–Cartan gravity admitting conformal motion have been explored in \cite{Sarkar2024EPJP}.

According to the seminal work in string/M-theory of \cite{Witten1996NPB, Seiberg1999JHEP}, the concept of spacetime quantization has emerged as a compelling framework wherein spacetime coordinates are treated as noncommuting operators on a D-brane. This noncommutativity is characterized by the relation $[x^\mu, x^\nu] = i \theta^{\mu\nu}$, where $\theta^{\mu\nu}$ is an antisymmetric matrix that encapsulates the fundamental discretization of spacetime. Such a framework replaces point-like structures with smeared objects in flat spacetime \cite{Smailagic2003JPA}. This smearing effect, which mitigates singularities, is mathematically implemented by substituting the Dirac delta function with a Gaussian distribution of minimal length \(\sqrt{\theta}\) \cite{Nicolini2006PLB}, which reflects the uncertainty encoded in the noncommutative structure of spacetime. \cite{Garattini2009PLB} extended this framework to wormhole geometries, investigating whether the smearing effects of noncommutativity can address the classical requirement of exotic matter, which involves a stress-energy tensor violating the null energy condition (NEC) \cite{Morris1988AJP}. This approach complements the aforementioned studies, providing an alternative perspective on the energy conditions required to sustain wormholes.

A crucial aspect of any wormhole geometry is its response to perturbations, typically analyzed through QNMs. These characteristic oscillations encode information about the stability of the structure and its interactions with surrounding fields. While QNMs have been extensively studied for black holes and classical wormholes \cite{Konoplya2003ATMP, Konoplya2005PRD, Konoplya2011RMP, Konoplya2016JCAP, Konoplya2018PLB, Konoplya2022PRL, Konoplya2010PRD, Berti2009CQG, Cardoso2018PRD, Cardoso2018PRL, Mamani2022EPJC, Batic2024CGG, Batic2024EPJC, Batic2024PRD,Batic2025EPJC}, their analysis in noncommutative geometry-inspired wormholes remains largely unexplored.

In this work, we investigate the QNMs of the noncommutative geometry-inspired wormholes under scalar, electromagnetic, and gravitational perturbations. By means of a spectral method, we compute the QNMs with high precision and analyze their dependence on the noncommutative parameter and other geometric characteristics. Our study not only provides insights into the stability and resonant dynamics of these wormholes but also establishes a robust computational framework for exploring exotic spacetimes in the context of quantum gravity.

To compute the QNMs of the noncommutative geometry-inspired wormholes, we adopt a spectral method similar to that employed in \cite{Mamani2022EPJC, Batic2024EPJC, Batic2024CGG, Batic2024PRD}. It is worth mentioning that the spectral method was rigorously validated against known reference values in \cite{Batic2024EPJC} and subsequently cross-checked with the third- and sixth-order WKB approaches in \cite{Batic2024CGG}. Beginning with the perturbation equations for massless scalar, electromagnetic, and gravitational fields, we reformulate them as an eigenvalue problem over a compact domain. This framework uses Chebyshev polynomials as basis functions, enabling fast convergence and high precision in the frequency domain. Our analysis does not detect the presence of purely imaginary QNMs, i.e. overdamped modes indicative of a rapid return to equilibrium without oscillation. Furthermore, we validate the robustness of our method by demonstrating that, for large values of the rescaled mass parameter, the QNMs of the noncommutative geometry-inspired wormholes converge to those of the classic Schwarzschild wormhole. This serves as a key benchmark, reinforcing the accuracy and reliability of our numerical approach.

The paper is organized as follows. Section II introduces the foundational aspects of noncommutative geometry-inspired wormholes, discussing the key properties of the line element and the governing equations for scalar, electromagnetic, and gravitational perturbations. Sections III and IV set up the application of the spectral method to the classic Schwarzschild wormhole and its noncommutative geometry-inspired counterpart. Here, the QNM boundary conditions are introduced as in \cite{BlazquezSalcedo2018PRD, Batic2024CGG}. In Section V, we describe the spectral method employed to compute the QNMs, including the mathematical formalism and numerical implementation. Section VI presents the results of our analysis, highlighting the dependence of QNMs on the wormhole parameters and comparing them with classical Schwarzschild wormhole results in the large mass limit. Finally, Section VI concludes the study with a summary of key findings and potential directions for future research. Supplementary materials and numerical codes used in this study are made available for transparency and reproducibility.

\section{THE NONCOMMUTATIVE GEOMETRY INSPIRED WORMHOLE: Metric and Equations of motion}

In this study, we investigate the behaviour of a massless scalar field denoted as $\psi$ within the spacetime of a noncommutative geometry-inspired wormhole. The manifold is described by a metric, expressed in natural units where $c = G_N = 1$, given by the following line element \cite{Garattini2009PLB, Nicolini2010CQG}
\begin{equation}\label{metric}
  ds^2 = -e^{2\Lambda(r)}dt^2+\frac{dr^2}{1-\frac{b(r)}{r}} + r^2d\vartheta^2 + r^2\sin^2{\vartheta}d\varphi^2, \quad \vartheta\in[0,\pi], \quad \varphi\in[0,2\pi).
\end{equation}
Here, $\Lambda(r)$ and $b(r)$ are defined as the lapse/redshift and shape functions, respectively, and are given by \cite{Garattini2009PLB, Nicolini2010CQG}
\begin{equation}
\Lambda(r)=0,\qquad b(r) = \frac{4M}{\sqrt{\pi}}\gamma\left(\frac{3}{2}, \frac{r^2}{4\theta}\right),
\end{equation} 
where
\begin{equation}
\gamma\left(\frac{3}{2}, \frac{r^2}{4\theta}\right) = \int_0^{r^2/4\theta} dt \, \sqrt{t} e^{-t},
\end{equation}
is the lower incomplete Gamma function, $M$ is the total mass of the gravitational object, and $\theta$ is a parameter encoding noncommutativity with the dimension of length squared. In the coordinate system $(t, r, \vartheta, \phi)$, the radial coordinate $r$ is constrained to $r > r_0$, where $r_0$ represents the position of the wormhole throat. Extending $r$ into the region $0 < r < b_0$ is not possible, as the metric tensor undergoes a signature change in this domain. However, it may still be feasible to define a second coordinate chart that is pairwise compatible with the original. This chart would overlap with the initial chart over an open subset of the manifold containing the wormhole throat and feature a smooth bijective transition function between the two. Such a construction would enable the description of the manifold region beyond the throat. This possibility will be explored further later in the discussion. In the following, it is convenient to rewrite \eqref{metric} in the form  
\begin{equation}\label{metric2}  
ds^2 = -dt^2 + \frac{dr^2}{f(r)} + r^2 d\vartheta^2 + r^2 \sin^2{\vartheta} \, d\varphi^2, \quad  
f(r) = 1 - \frac{4M}{\sqrt{\pi}r}\gamma\left(\frac{3}{2}, \frac{r^2}{4\theta}\right),  
\end{equation}  
as this representation, resembling the metric of a noncommutative geometry-inspired Schwarzschild black hole, simplifies the analysis of the throat location. Specifically, the procedure can be carried out in a manner analogous to the study of the event and Cauchy horizons in the aforementioned black hole \cite{Nicolini2006PLB, Batic2024EPJC}. If we introduce the rescaling
\begin{equation}
  x = \frac{r}{2M}, \qquad \mu = \frac{M}{\sqrt{\theta}},
\end{equation}
and use the identities \cite{Abramowitz1972}
\begin{equation}\label{ids}
  \gamma\left(\frac{3}{2}, w^2\right) = \frac{1}{2} \gamma\left(\frac{1}{2}, w^2\right) - w e^{-w^2}, \quad \gamma\left(\frac{1}{2}, w^2\right) = \sqrt{\pi} \, \mathrm{erf}(w),
\end{equation}
we can rewrite $f(r)$ in \eqref{metric2} as
\begin{equation}\label{f}
  f(x) = 1 - \frac{\mathrm{erf}(\mu x)}{x} + \frac{2\mu}{\sqrt{\pi}} e^{-\mu^2 x^2},
\end{equation}
where \(\mathrm{erf}(\cdot)\) denotes the error function. The insights drawn from Figure~\ref{fig0}, where $f$ is plotted as a function of the rescaled radial variable $x$, are as follows
\begin{enumerate}
\item 
\underline{Case $\mu < \mu_e = 1.904119076\ldots$}: There are no real roots, and consequently, no throats exist.  
\item 
\underline{Case $\mu = \mu_e$}: Two coinciding throats are present at $x_e = 0.7936575898$.  \item 
\underline{Case $\mu > \mu_e$}: There are two distinct throats located at $x_\pm$. The throat is defined to be at $x_0 = x_+$, as this root increases with $\mu$, while $x_- \to 0$. This definition ensures the throat location remains consistent with the expectation that it should grow larger as the mass parameter increases.  
\end{enumerate}
It is worth noting that while the non-extreme cases considered in this work satisfy the flare-out condition \cite{Garattini2009PLB} , thereby qualifying as traversable wormholes, the extreme case ($\mu = \mu_e$) represents a degenerate configuration where the two throats coincide. This extreme configuration does not satisfy the flare-out condition, and thus does not constitute a traversable wormhole. Although its physical interpretation remains unclear, it may represent a soliton-like gravitational structure or a geometric defect within noncommutative spacetime, deserving further investigation in future work.

\begin{figure}
\includegraphics[scale=0.35]{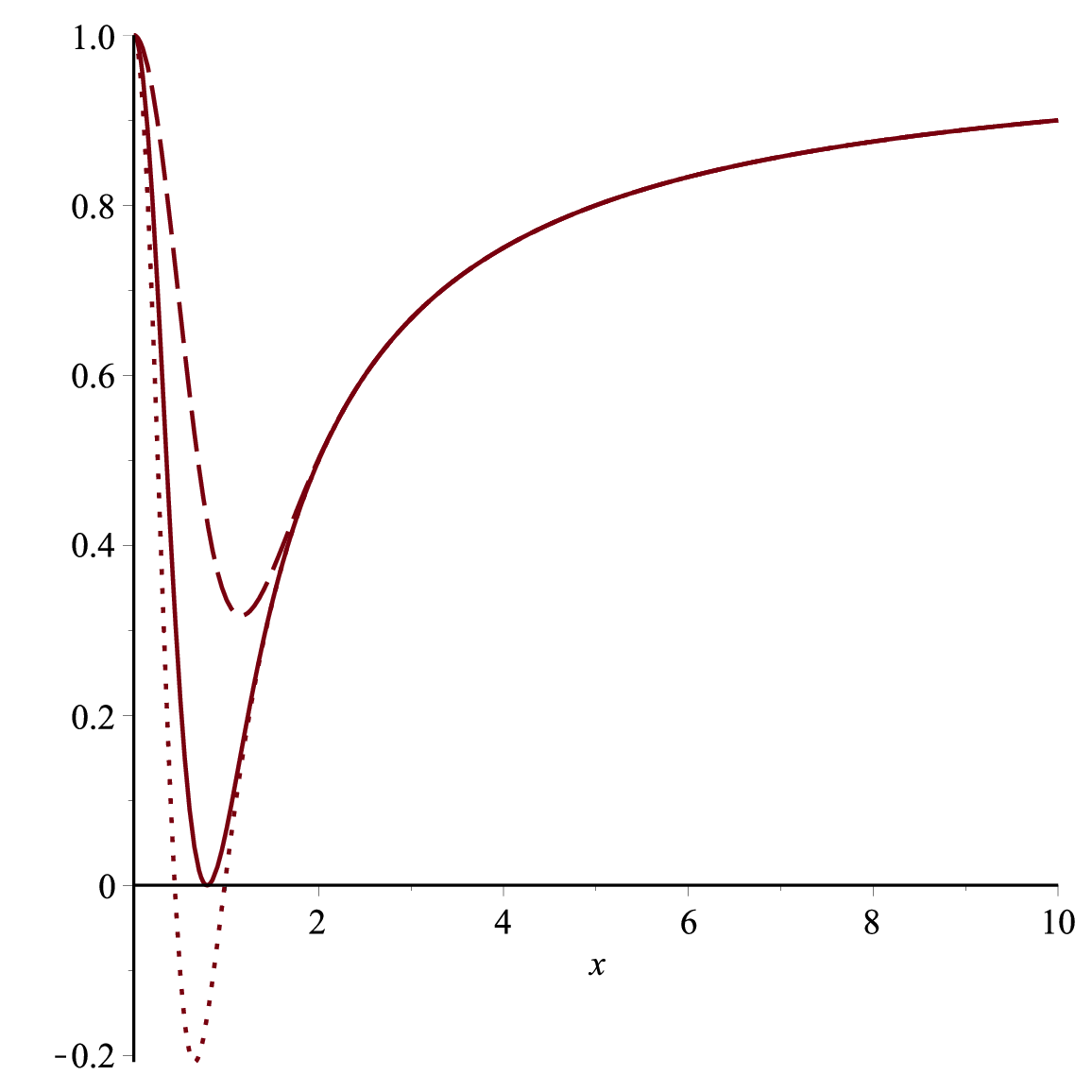}
\caption{\label{fig0}
Plot of the metric coefficient $g^{xx} = f(x)$, as defined in \eqref{f}. The extremal case occurs at $\mu = \mu_e = 1.904119076\ldots$ (solid line), where two coinciding throats are located at $x_e = 0.7936575898\ldots$. For $\mu > \mu_e$, a non-extremal wormhole with two distinct throats is present (dotted line shown for $\mu = 2.3$). For $\mu < \mu_e$, the absence of throats indicates that the gravitational object is not a wormhole (dashed line shown for $\mu = 1.3$).}
\end{figure}

Finally, it is important to note that when $\mu \gg 1$, equation \eqref{f} simplifies to  
\begin{equation}  
    f(x) = 1 - \frac{1}{x} + \mathcal{O}\left(e^{-\mu^2 x^2}\right).  
\end{equation}  
This indicates that, in the limit of large mass parameters (or equivalently, as $r/\sqrt{\theta} \to \infty$), the classical counterpart of the manifold described by \eqref{metric2} is a Schwarzschild wormhole rather than a Morris-Thorne wormhole. This observation is significant for two reasons. First, it justifies taking $x_+$ as the position of the throat instead of $x_-$. In the regime where $\mu \gg 1$, $x_- \to 0$ while $x_+ \to 1$, making $x_+$ the physically meaningful throat location. Second, it suggests that the QNMs of the noncommutative wormhole should closely approximate those of a Schwarzschild wormhole when the mass parameter is large. This serves as a useful benchmark for validating the accuracy of the QNM results obtained using the spectral method.

A study by \cite{Kim2008PTPS} demonstrated that within the spacetime described by the metric \eqref{metric}, and in the coordinate system $(t, r, \vartheta, \varphi)$ with $r > r_0$, the governing equation for a massless Klein-Gordon field can be expressed in a specific form. The field is assumed to have a time dependence of the form $e^{-i\omega t}$ and an angular dependence described by spherical harmonics. For $\ell = 0, 1, 2, \ldots$ and different types of perturbations ($s = 0$ for massless scalar perturbations, $s = 1$ for electromagnetic perturbations, and $s = 2$ for vector-type gravitational perturbations), the field equation takes the form  
\begin{equation}\label{ODE01}  
  \sqrt{f(r)}\frac{d}{dr}\left(\sqrt{f(r)}\frac{d\psi_{\omega\ell s}}{dr}\right) + \left[\omega^2 - U_{s,\ell}(r)\right]\psi_{\omega\ell s}(r) = 0,  
\end{equation}  
where 
\begin{equation}
  U_{s,\ell}(r)=\frac{\ell(\ell+1)}{r^2}+\frac{(1-s)(1+2s)}{2r}\frac{df}{dr}
\end{equation}
represents the effective potential for the perturbation of type $s$ and multipole moment $\ell$. At this point, introducing the rescaling $z = x/x_0$ leads to the following differential equation
\begin{eqnarray}
  &&\sqrt{f(z)}\frac{d}{dz}\left(\sqrt{f(z)}\frac{d\psi_{\Omega\ell s}}{dz}\right)+\left[x_0^2 \Omega^2-V_{s,\ell}(z)\right]\psi_{\Omega\ell s}(z)=0,\quad\Omega=2M\omega,\label{ODEor}\\
  &&V_{s,\ell}(z)=\frac{\ell(\ell+1)}{z^2}+\frac{(1-s)(1+2s)}{2z}\frac{df}{dz},\quad z>1,\\
  &&f(z)= 1 - \frac{\mathrm{erf}\left(\mu x_0 z\right)}{x_0 z} + \frac{2\mu}{\sqrt{\pi}} e^{-\mu^2 x_0^2 z^2}.
\end{eqnarray}
The following section focuses on extracting the QNMs of the Schwarzschild wormhole using the spectral method.

\section{Extraction of QNMs for the Schwarzschild Wormhole}
\label{WS}

In this case, $x_0 = 1$, and thus $z = x$. Additionally, the metric coefficient $g^{xx}$  simplifies to $f(x) = 1 - x^{-1}$. As a consequence, equation \eqref{ODEor} reduces to  
\begin{eqnarray}\label{ODEorS}  
    &&\sqrt{1 - \frac{1}{x}} \frac{d}{dx} \left(\sqrt{1 - \frac{1}{x}} \frac{d\psi_{\Omega\ell s}}{dx}\right) + \left[\Omega^2 - V_{(S)}(x)\right]\psi_{\Omega\ell s}(x) = 0, \quad \Omega = 2M\omega, \\  
    &&V_{(S)}(x) = \frac{\ell(\ell+1)}{x^2} + \frac{(1-s)(1+2s)}{2x^3}, \quad x > 1,  
\end{eqnarray}  
where, for notational convenience, the subscripts $s$ and $\ell$ in the effective potential have been replaced with the subscript $S$, representing the Schwarzschild case. To classify the singularities of equation \eqref{ODEorS}, it is useful to rewrite it as follows
\begin{eqnarray}
  &&\frac{d^2\psi_{\Omega\ell s}}{dx^2}+p(x)\frac{d\psi_{\Omega\ell s}}{dx}+q(x)\psi_{\Omega\ell s}(x)=0,\quad x>1,\label{ODEzn}\\
  &&p(x)=-\frac{1}{2x}+\frac{1}{2(x-1)},\\
  &&q(x)=\Omega^2+\frac{2\ell(\ell+1)+(1-s)(1+2s)}{2x}+\frac{2\Omega^2-2\ell(\ell+1)-(1-s)(1+2s)}{2(x-1)}+\frac{(1-s)(1+2s)}{2x^2}.
\end{eqnarray}
The differential equation \eqref{ODEzn} reveals the presence of two regular singular points at $x = 0$ and $x = 1$. However, only the singularity at $x = 1$ is relevant, as the other lies outside the valid range of the variable $x$. Additionally, the point at infinity is identified as an irregular singular point of rank 1 \cite{Bender1999}. The exponents of the regular singularity at $x = 1$ can be determined using Frobenius theory. Specifically, they are obtained by solving the indicial equation
\begin{equation}  
    \rho(\rho - 1) + p_0\rho + q_0 = 0,  
\end{equation}  
where  
\begin{equation}  
    p_0 = \lim_{x \to 1^+}(x - 1)p(x) = \frac{1}{2}, \quad  
    q_0 = \lim_{x \to 1^+}(x - 1)^2 q(x) = 0.  
\end{equation}  
The roots of the indicial equation are $\rho_1 = 0$ and $\rho_2 = 1/2$.  To proceed, the derivation of the QNM boundary conditions and the implementation of the spectral method are organized into two separate subsections, which correspond to the different cases arising from the distinct exponents at the regular singular point $x = 1$.  

\subsection{Case I}\label{subA}

Let us pick the first one, which is equivalent to the requirement that $\Psi_{\Omega\ell s}$ remains finite at the throat. The asymptotic behaviour of the solutions to equation \eqref{ODEzn} can be deduced using the method outlined in \cite{Olver1994MAA}. For this purpose, we observe that
\begin{equation}
  p(x) = \sum_{\kappa=0}^\infty\frac{\mathfrak{f}_\kappa}{x^k} = \mathcal{O}\left(\frac{1}{x^2}\right), \qquad
  q(x) = \sum_{\kappa=0}^\infty\frac{\mathfrak{g}_\kappa}{x^k}=\Omega^2+\frac{\Omega^2}{x}+\mathcal{O}\left(\frac{1}{x^2}\right).
\end{equation}
Given that at least one of the coefficients $\mathfrak{f}_0$, $\mathfrak{g}_0$, $\mathfrak{g}_1$ is nonzero, a formal asymptotic solution to \eqref{ODEzn} is represented by \cite{Olver1994MAA}
\begin{equation}
  \psi^{(j)}_{\Omega\ell s}(x) = x^{\mu_j}e^{\lambda_j x}\sum_{\kappa=0}^\infty\frac{a_{\kappa,j}}{x^\kappa}, \qquad j \in \{1,2\},
\end{equation}
where $\lambda_1$, $\lambda_2$, $\mu_1$ and $\mu_2$ are the roots of the characteristic equations
\begin{equation}
   \lambda^2 + \mathfrak{f}_0\lambda + \mathfrak{g}_0 = 0,\quad
   \mu_j = -\frac{\mathfrak{f}_1\lambda_j + \mathfrak{g}_1}{\mathfrak{f}_0 + 2\lambda_j}.
\end{equation}
A straightforward computation shows that $\lambda_\pm = \pm i\Omega$ and $\mu_\pm = \pm i\Omega/2$. As a result, we impose that the radial field exhibits outward radiation in the limit of $x \to +\infty$, that is
\begin{equation}
  \psi_{\Omega\ell s}\underset{{x\to +\infty}}{\longrightarrow} x^{\frac{i\Omega}{2}}e^{i\Omega x},
\end{equation}
while remaining finite at the throat ($\rho_1=0$). As a preliminary step in applying the Spectral Method, we introduce the following ansatz that entails the correct behaviour at the throat and towards positive space-like infinity
\begin{equation}
\psi_{\Omega\ell s}(x) = x^{\frac{i\Omega}{2}} e^{i\Omega(x-1)}U_{\Omega\ell s}(x).
\end{equation}
Additionally, we impose the condition that $U_{\Omega\ell s}(x)$ remains regular as $x \to 1^{+}$ and as $x \to +\infty$. Then, \eqref{ODEzn} becomes
\begin{eqnarray}
  &&Q_2(x)\frac{d^2 U_{\Omega\ell s}}{dx^2}+Q_1(x)\frac{dU_{\Omega\ell s}}{dx}+Q_0(x)U_{\Omega\ell s}(x)=0,\quad x>1,\label{ODEznt}\\
  &&Q_2(x)=1,\quad
  Q_1(x)=\frac{1+4i\Omega(x-1)(x-1/2)}{2x(x-1)},\\
  &&Q_0(x)=\frac{[3\Omega^2-4\ell(\ell+1)]x+\Omega^2+3i\Omega-2(1-s)(1+2s)}{4x^2(x-1)}.
\end{eqnarray}
In view of the application of the Spectral Method, we need to map the interval $(1, +\infty)$ onto the interval $(-1,1)$. This transformation is accomplished using $y = 1 - 2x^{-1}$, which sends $+\infty$ to $1$, and $1$ to $-1$. Consequently, the transformed differential equation is represented as follows
\begin{equation}\label{ODEm2}
  S_2(y)\ddot{U}_{\Omega\ell s}(y) + S_1(y)\dot{U}_{\Omega\ell s}(y) + S_0(y)U_{\Omega\ell s}(y) = 0,
\end{equation}
where
\begin{equation}
  S_2(y) = \frac{(1-y)^4}{4}, \quad
  S_1(y) = i\Omega E_1(y)+F_1(y), \quad
  S_0(y) = \Omega^2\Sigma_2(y) + i\Omega\Sigma_1(y) + \Sigma_0(y) 
\end{equation}
with
\begin{eqnarray}
  E_1(y) &=& -\frac{(y-5)(1-y)^2}{4},\quad
  F_1(y) = -\frac{(5y+3)(1-y)^3}{8(1+y)}, \\
\Sigma_2(y) &=&-\frac{(1-y)^2(y-7)}{16(1+y)},\quad
  \Sigma_1(y) = \frac{3(1-y)^3}{16(1+y)},\quad
  \Sigma_0(y) = \frac{(1-y)^2}{8(1+y)}\left[(2s+1)(s-1)(1-y)-4\ell(\ell+1)\right].
\end{eqnarray}

The direct analysis of the coefficient functions $S_i(y)$ shows that both $S_1(y)$ and $S_0(y)$ possess a simple pole at $y = 1$. Additionally, all $S_i(y)$ share a common second-order zero at $y = 1$. Therefore, in view of the application of the Spectral Method, we need to multiply equation \eqref{ODEm2} by $(1+y)/(1-y)^2$. As a consequence, we obtain the following differential equation
\begin{equation}\label{ODEhtfinal}
  M_2(y)\ddot{U}_{\Omega\ell }(y) + M_1(y)\dot{U}_{\Omega\ell s}(y) + M_0(y)U_{\Omega\ell s}(y) = 0,
\end{equation}
where
\begin{equation}\label{S210fin}
  M_2(y)=\frac{1}{4}(1+y)(1-y)^2,\quad
  M_1(y)=i\Omega\widehat{N}_1(y)+\widehat{N}_0(y),\quad
  M_0(y)=\Omega^2\widehat{C}_2(y)+i\Omega\widehat{C}_1(y)+\widehat{C}_0(y)
\end{equation}
with
\begin{eqnarray}
  N_1(y)&=&-\frac{1}{4}(1+y)(y-5),\quad 
  N_0(y)=-\frac{1}{8}(1-y)(5y+3),\label{N0fin}\\
  C_2(y)&=&\frac{1}{16}(7-y),\quad
  C_1(y)=\frac{3}{16}(1-y),\quad
  C_0(y)=\frac{1}{8}(1+2s)(s-1)(1-y)-\frac{\ell(\ell+1)}{2}.\label{C210}
\end{eqnarray}
It can be easily verified with Maple that
\begin{eqnarray}
    &&\lim_{y\to 1^{-}}M_2(y)=0=\lim_{y\to -1^{+}}M_2(y),\\
    &&\lim_{y\to 1^{-}}M_1(y)=2i\Omega,\quad
    \lim_{y\to -1^{+}}M_1(y)=\frac{1}{2},\\
    &&\lim_{y\to 1^{-}}M_0(y)=\frac{3}{8}\Omega^2-\frac{\ell(\ell+1)}{2},\quad
     \lim_{y\to -1^{+}}M_0(y)=\frac{\Omega^2}{2}+\frac{3}{8}i\Omega+\frac{1}{4}(2s+1)(s-1)-\frac{\ell(\ell+1)}{2}.
\end{eqnarray} 
In the final step, as we prepare to apply the spectral method, we transform the differential equation \eqref{ODEhtfinal} into the following form
\begin{equation}\label{TSCH}
  L_0\left[U_{\Omega\ell s}, \dot{U}_{\Omega\ell s}, \ddot{U}_{\Omega\ell s}\right] +  i L_1\left[U_{\Omega\ell s}, \dot{U}_{\Omega\ell s}, \ddot{U}_{\Omega\ell s}\right]\Omega +  L_2\left[U_{\Omega\ell s}, \dot{U}_{\Omega\ell s}, \ddot{U}_{\Omega\ell s}\right]\Omega^2 = 0.
\end{equation}
Here, we have
\begin{eqnarray}
  L_0\left[U_{\Omega\ell s}, \dot{U}_{\Omega\ell s}, \ddot{U}_{\Omega\ell s}\right] &=&L_{00}(y)U_{\Omega\ell s} + L_{01}(y)\dot{U}_{\Omega\ell s} + L_{02}(y)\ddot{U}_{\Omega\ell s},\label{L0none}\\
  L_1\left[U_{\Omega\ell s}, \dot{U}_{\Omega\ell s}, \ddot{U}_{\Omega\ell s}\right] &=& L_{10}(y)U_{\Omega\ell s} + L_{11}(y)\dot{U}_{\Omega\ell s} + L_{12}(y)\ddot{U}_{\Omega\ell s}, \label{L1none}\\
  L_2\left[U_{\Omega\ell s}, \dot{U}_{\Omega\ell s}, \ddot{U}_{\Omega\ell s}\right] &=& L_{20}(y)U_{\Omega\ell s} + L_{21}(y)\dot{U}_{\Omega\ell s} + L_{22}(y)\ddot{U}_{\Omega\ell s}.\label{L2none}
\end{eqnarray}
We direct the reader to Table~\ref{tableAlt} where we summarised the $L_{ij}$ terms presented in equations \eqref{L0none} through \eqref{L2none}, along with their respective limiting values at $y = \pm 1$.

\begin{table}
\caption{Definitions of the coefficients $L_{ij}$ and their corresponding behaviours at the endpoints of the interval $-1 \leq y \leq 1$.}
\begin{center}
\begin{tabular}{ | c | c | c | c | c | c | c | c }
\hline
$(i,j)$  & $\displaystyle{\lim_{y\to -1^+}}L_{ij}$  & $L_{ij}$ & $\displaystyle{\lim_{y\to 1^-}}L_{ij}$  \\ \hline
$(0,0)$ &  $\frac{(2s+1)(s-1)}{4}-\frac{\ell(\ell+1)}{2}$          & $C_0$                  & $-\frac{\ell(\ell+1)}{2}$\\ \hline
$(0,1)$ &  $\frac{1}{2}$            & $N_0$                  & $0$\\ \hline
$(0,2)$ &  $0$            & $M_2$                  & $0$\\ \hline 
$(1,0)$ &  $\frac{3}{8}$            & $C_1$                  & $0$\\ \hline 
$(1,1)$ &  $0$            & $N_1$                  & $2$\\ \hline 
$(1,2)$ &  $0$            & $0$                              & $0$\\ \hline 
$(2,0)$ &  $\frac{1}{2}$            & $C_2$                  & $\frac{3}{8}$\\ \hline
$(2,1)$ &  $0$            & $0$                    & $0$\\ \hline
$(2,2)$ &  $0$            & $0$                    & $0$\\ \hline
\end{tabular}
\label{tableAlt}
\end{center}
\end{table} 

\subsection{Case II}
\label{subB}

In this case, we select the Frobenius solution corresponding to the exponent $\rho = 1/2$, which imposes the condition $\psi_{\Omega\ell s}(1) = 0$. Using the following ansatz, which captures the correct behaviour at the throat and asymptotically at space-like infinity ($x \to +\infty$) 
\begin{equation}
  \psi_{\Omega\ell\epsilon}(x) = \sqrt{1 - \frac{1}{x}} \, x^{\frac{i\Omega}{2}} e^{i\Omega(x-1)} U_{\Omega\ell s}(x),
\end{equation}  
and imposing the condition that $U_{\Omega\ell s}(x)$ remains regular as $x \to 1^+$ and $x \to +\infty$, equation \eqref{ODEzn} can be rewritten in the form of \eqref{ODEznt}, with the coefficients $Q_i$ replaced by
\begin{eqnarray}
  \widetilde{Q}_2(x)&=&1,\quad
  \widetilde{Q}_1(x)=\frac{4i\Omega x^2-2i\Omega x-2i\Omega+3}{2x(x-1)},\\
  \widetilde{Q}_0(x)&=&\frac{[3\Omega^2+4i\Omega-4\ell(\ell+1)]x+\Omega^2+5i\Omega+4s^2-2s-6}{4x^2(x-1)}.    
\end{eqnarray}
The transformation $y = 1 - 2x^{-1}$ maps the interval $(1, +\infty)$ onto the entire real line. Under this mapping, the resulting equation keeps the same form as \eqref{ODEm2}, except that the coefficients $S_i$ are replaced by
\begin{equation}
  \widetilde{S}_2(y)=\frac{(1-y)^4}{4},\quad
  \widetilde{S}_1(y)=i\Omega \widetilde{E}_1(y)+\widetilde{F}_1(y),\quad
  \widetilde{S}_0(y) = \Omega^2\widetilde{\Sigma}_2(y) + i\Omega\widetilde{\Sigma}_1(y) + \widetilde{\Sigma}_0(y) 
\end{equation}
with
\begin{eqnarray}
  \widetilde{E}_1(y)&=&-\frac{(y-5)(1-y)^2}{4},\quad
  \widetilde{F}_1(y)=-\frac{(1+7y)(1-y)^3}{8(1+y)},\quad
  \widetilde{\Sigma}_2(y)=\frac{(1-y)^2(7-y)}{16(1+y)},\\
  \widetilde{\Sigma}_1(y)&=&\frac{(1-y)^2(13-5y)}{16(1+y)},\quad
  \widetilde{\Sigma}_0(y)=\frac{(1-y)^2}{8(1+y)}\left[(s+1)(2s-3)(1-y)-4\ell(\ell+1)\right].
\end{eqnarray}
An examination of the coefficient functions $\widetilde{S}_i(y)$ reveals that all coefficients have a common second-order zero at $y = 1$, while both $\widetilde{S}_1(y)$ and $\widetilde{S}_0(y)$ exhibit a simple pole at $y = -1$. To enable the application of the Spectral Method, the corresponding differential equation must be multiplied by $(1 + y)/(1 - y)^2$. This modification produces a differential equation analogous to \eqref{ODEhtfinal}, with the coefficients $M_i(y)$ replaced by
\begin{equation}\label{S210fin2}
  \widetilde{M}_2(y)=\frac{(1+y)(1-y)^2}{4},\quad
  \widetilde{M}_1(y)=i\Omega\widetilde{N}_1(y)+\widetilde{N}_0(y),\quad
\widetilde{M}_0(y)=\Omega^2\widetilde{C}_2(y)+i\Omega\widetilde{C}_1(y)+\widetilde{C}_0(y)
\end{equation}
with
\begin{eqnarray}
\widetilde{N}_1(y)&=&\frac{(1+y)(5-y)}{4},\quad 
\widetilde{N}_0(y)=-\frac{(1-y)(1+7y)}{8},\label{N0fin2}\\
\widetilde{C}_2(y)&=&\frac{7-y}{16},\quad
\widetilde{C}_1(y)=\frac{13-5y}{16},\quad
\widetilde{C}_0(y)=\frac{1}{8}(1+s)(2s-3)(1-y)-\frac{\ell(\ell+1)}{2}.\label{C2102}
\end{eqnarray}
It is straightforward to check that
\begin{eqnarray}
    &&\lim_{y\to 1^{-}}\widetilde{M}_2(y)=0=\lim_{y\to -1^{+}}\widehat{M}_2(y),\\
    &&\lim_{y\to 1^{-}}\widetilde{M}_1(y)=2i\Omega,\quad
    \lim_{y\to -1^{+}}\widetilde{M}_1(y)=\frac{3}{2},\\
    &&\lim_{y\to 1^{-}}\widetilde{M}_0(y)=\frac{3\Omega^2}{8}+i\frac{\Omega}{2}-\frac{\ell(\ell+1)}{2},\quad
     \lim_{y\to -1^{+}}\widetilde{M}_0(y)=\frac{\Omega^2}{2}+i\frac{9\Omega}{8}+\frac{1}{4}(s+1)(2s-3)-\frac{\ell(\ell+1)}{2}.
\end{eqnarray} 
At this stage, we can adopt the same approach outlined in the previous subsection to transform the corresponding differential equation into the form of \eqref{TSCH}. For clarity, we refer the reader to Table~\ref{tableAlt2}, which summarizes the $\widetilde{L}_{ij}$ terms that replace the $L_{ij}$ terms in equations \eqref{L0none} through \eqref{L2none}, along with their respective limiting values at \(y = \pm 1\).

\begin{table}
\caption{Definitions of the coefficients $\widetilde{L}_{ij}$ and their corresponding behaviours at the endpoints of the interval $-1 \leq y \leq 1$.}
\begin{center}
\begin{tabular}{ | c | c | c | c | c | c | c | c }
\hline
$(i,j)$  & $\displaystyle{\lim_{y\to -1^+}}\widetilde{L}_{ij}$  & $\widetilde{L}_{ij}$ & $\displaystyle{\lim_{y\to 1^-}}\widetilde{L}_{ij}$  \\ \hline
$(0,0)$ &  $\frac{(2s-3)(s+1)}{4}-\frac{\ell(\ell+1)}{2}$          & $\widetilde{C}_0$                  & $-\frac{\ell(\ell+1)}{2}$\\ \hline
$(0,1)$ &  $\frac{3}{2}$            & $\widetilde{N}_0$                  & $0$\\ \hline
$(0,2)$ &  $0$            & $\widetilde{M}_2$                  & $0$\\ \hline 
$(1,0)$ &  $\frac{9}{8}$            & $\widetilde{C}_1$                  & $\frac{1}{2}$\\ \hline 
$(1,1)$ &  $0$            & $\widetilde{N}_1$                  & $2$\\ \hline 
$(1,2)$ &  $0$            & $0$                              & $0$\\ \hline 
$(2,0)$ &  $\frac{1}{2}$            & $\widetilde{C}_2$                  & $\frac{3}{8}$\\ \hline
$(2,1)$ &  $0$            & $0$                    & $0$\\ \hline
$(2,2)$ &  $0$            & $0$                    & $0$\\ \hline
\end{tabular}
\label{tableAlt2}
\end{center}
\end{table}

\section{Extraction of QNMs for the noncommutative geometry - inspired Wormhole}
\label{NCW}

To establish the QNM boundary conditions at the throat and at infinity, we first need to determine the asymptotic behaviour of the radial solution $\psi_{\Omega \ell s}$ as $z \to 1^+$ and $z \to +\infty$. From this asymptotic analysis, we can then extract the QNM boundary conditions.  Later, as a consistency check, we will verify that for $\mu\gg 1$, the QNMs of the noncommutative geometry-inspired wormhole go over into the corresponding QNMs of the Schwarzschild wormhole. We conduct our analysis by examining the behaviour of the radial field in two separate regions.

\begin{enumerate}
\item 
{\underline{Asymptotic behavior as $z\to 1^+$}}: Given that $z = 1$ is a simple zero of $f(z)$, we can express it in the form $f(z) = (z - 1)h(z)$, where $h$ is an analytic function at $z = 1$ and satisfies $h(1) = f'(1) \neq 0$. Here, the prime symbol denotes differentiation with respect to $z$. This formulation allows us to rewrite equation (\ref{ODEor}) in the following form
\begin{eqnarray}
    &&\frac{d^2\psi_{\Omega \ell s}}{dz^2} + \mathfrak{p}(z) \frac{d\psi_{\Omega \ell s}}{dz} + \mathfrak{q}(z)\psi_{\Omega \ell s}(z) = 0, \label{ODEZW} \\
    &&\mathfrak{p}(z)=\frac{f^{'}(z)}{2f(z)} = \frac{1}{2(z - 1)} + \frac{h'(z)}{h(z)}, \\
    &&\mathfrak{q}(z)=\frac{x_0^2\Omega^2-V_{s,\ell}(z)}{f(z)} = \frac{x_0^2 \Omega^2}{(z - 1) h(z)} - \frac{1}{(z - 1)h(z)} \left[\frac{\ell(\ell + 1)}{z^2} + \frac{(1-s)(1+2s)}{2z}\frac{df}{dz}\right].
\end{eqnarray}
Since both $\mathfrak{p}$ and $\mathfrak{q}$ have a simple pole at $z = 1$, this point can be classified as a regular singular point of (\ref{ODEZW}) by Frobenius theory \cite{Ince1956}. Consequently, we can construct solutions in the form
\begin{equation}
    \psi_{\Omega \ell s}(z) = (z - 1)^\rho \sum_{\kappa = 0}^\infty a_\kappa (z - 1)^\kappa.
\end{equation}
The leading behavior at $z = 1$ is represented by the term $(z - 1)^\rho$, where $\rho$ is determined by the indicial equation
\begin{equation}\label{indicial}
    \rho(\rho - 1) + \mathfrak{p}_0 \rho + \mathfrak{q}_0 = 0,
\end{equation}
with
\begin{equation}
    \mathfrak{p}_0 = \lim_{z \to 1} (z - 1)\mathfrak{p}(z) = \frac{1}{2}, \qquad
    \mathfrak{q}_0 = \lim_{z \to 1} (z - 1)^2\mathfrak{q}(z) =0.
\end{equation}
The roots of (\ref{indicial}) are $\rho_1 = 0$ and $\rho_2=1/2$.
\item
{\underline{Asymptotic behaviour as $z\to+\infty$}}: It can be obtained by using the method described in \cite{Olver1994MAA}. To proceed, we observe that
\begin{equation}
    \mathfrak{p}(z) = \sum_{\kappa=0}^\infty \frac{\mathfrak{f}_\kappa}{z^k} = \mathcal{O}\left(\frac{1}{z^2}\right), \qquad
    \mathfrak{q}(z) = \sum_{\kappa=0}^\infty \frac{\mathfrak{g}_\kappa}{z^k} = x_0^2 \Omega^2 + \frac{x_0 \Omega^2}{z} + \mathcal{O}\left(\frac{1}{z^2}\right).
\end{equation}
Given that at least one of the coefficients \( \mathfrak{f}_0 \), \( \mathfrak{g}_0 \), or \( \mathfrak{g}_1 \) is non-zero, a formal solution to (\ref{ODEZW}) is given by \cite{Olver1994MAA}
\begin{equation}
    \psi^{(j)}_{\Omega \ell s}(z) = z^{\mu_j} e^{\lambda_j z} \sum_{\kappa=0}^\infty \frac{a_{\kappa, j}}{z^\kappa}, \qquad j \in \{1,2\},
\end{equation}
where \( \lambda_1 \), \( \lambda_2 \), \( \mu_1 \), and \( \mu_2 \) are the roots of the characteristic equations
\begin{equation}
    \lambda^2 + \mathfrak{f}_0 \lambda + \mathfrak{g}_0 = 0, \quad
    \mu_j = -\frac{\mathfrak{f}_1 \lambda_j + \mathfrak{g}_1}{\mathfrak{f}_0 + 2 \lambda_j}.
\end{equation}
A straightforward calculation shows that $\lambda_\pm = \pm i x_0 \Omega$ and $\mu_\pm = \pm i \Omega/2$. Consequently, the QNM boundary condition at spatial infinity can be expressed as
\begin{equation}\label{QNMBCzinf}
    \psi_{\Omega \ell s} \underset{{z \to +\infty}}{\longrightarrow} z^{\frac{i\Omega}{2}} e^{i x_0 \Omega z}.
\end{equation}
\end{enumerate}
At this stage, we can follow the same procedure as for the classic Schwarzschild wormhole or the Morris-Thorne wormhole \cite{Batic2024CGG} to establish two sets of QNM boundary conditions.

\subsection{Case A}
\label{subANC}

We impose that the radial field exhibits outward radiation in the limit of $z \to +\infty$, that is
\begin{equation}
  \psi_{\Omega\ell s}\underset{{z\to +\infty}}{\longrightarrow} z^{\frac{i\Omega}{2}}e^{i x_0\Omega z},
\end{equation}
while remaining finite at the throat ($\rho_1=0$). As a preliminary step in applying the Spectral Method, we introduce the following ansatz that simultaneously captures the correct behaviour at the throat and towards positive space-like infinity
\begin{equation}
\psi_{\Omega\ell s}(z) = z^{\frac{i\Omega}{2}} e^{i x_0\Omega(z-1)}\Phi_{\Omega\ell s}(z).
\end{equation}
Additionally, we impose the condition that $\Phi_{\Omega\ell s}(z)$ remains regular as $z \to 1^{+}$ and as $z \to +\infty$. Then, \eqref{ODEZW} becomes
\begin{equation}\label{ODEznoneA}
    \mathfrak{P}_2(z) \Phi''_{\Omega \ell s}(z) + \mathfrak{P}_1(z) \Phi'_{\Omega \ell s}(z) + \mathfrak{P}_0(z) \Phi_{\Omega \ell s}(z) = 0
\end{equation}
with
\begin{eqnarray}
    \mathfrak{P}_2(z)&=&z^2f(z),\quad \mathfrak{P}_1(z)=\frac{z^2}{2}f^{'}(z)+i\Omega z(2x_0 z+1)f(z),\\
    \mathfrak{P}_0(z)&=&\left[x_0^2 z^2(1-f(z))-\left(x_0 z+\frac{1}{4}\right)f(z)\right]\Omega^2+\frac{i\Omega}{2}\left[z\left(x_0 z+\frac{1}{2}\right)f^{'}(z)-f(z)\right]-z^2 V_{s,\ell}(z).\\
\end{eqnarray}
At this stage, it is convenient to apply the transformation $z = 2/(1-y)$, which maps the point at infinity and the throat to $y = 1$ and $y = -1$, respectively. In the following, a dot stands for differentiation with respect to the new variable $y$. As a result, equation (\ref{ODEznoneA}) becomes
\begin{equation}\label{ODEynoneA}
    \mathfrak{S}_2(y)\ddot{\Phi}_{\Omega\ell s}(y) + \mathfrak{S}_1(y)\dot{\Phi}_{\Omega\ell s}(y) + \mathfrak{S}_0(y)\Phi_{\Omega\ell s}(y) = 0,
\end{equation}
where
\begin{eqnarray}
  \mathfrak{S}_2(y) &=&(1-y)^2 f(y), \label{S2ononeA} \\
  \mathfrak{S}_1(y) &=&-2(1-y)f(y)+\frac{(1-y)^2}{2}\dot{f}(y)+i\Omega(4x_0+1-y)f(y),\label{S1ononeA}\\
  \mathfrak{S}_0(y) &=&\left[\frac{4x_0^2(1-f(y))}{(1-y)^2}-\left(\frac{2x_0}{1-y}+\frac{1}{4}\right)f(y)\right]\Omega^2+\frac{i\Omega}{2}\left[\left(2x_0+\frac{1-y}{2}\right)\dot{f}(y)-f(y)\right]-\frac{4V_{s,\ell}(y)}{(1-y)^2}. \label{S0ononeA}
\end{eqnarray}
In addition, we also require that $\Phi_{\Omega\ell s}(y)$ remains regular at $y=\pm 1$. As a consequence of the transformation introduced above, we have
\begin{eqnarray}
f(y)&=& 1 - \frac{1-y}{2x_0}\mbox{erf}\left(\frac{2\mu x_0}{1-y}\right) + \frac{2\mu}{\sqrt{\pi}}e^{-\frac{4\mu^2 x_0^2}{(1-y)^2}},\label{fyA}\\
V_{s,\ell}(y) &=&\frac{(1-y)^2}{4}\left[\ell(\ell+1)+\frac{(1-s)(1+2s)}{2}(1-y)\dot{f}(y)\right],\label{fvA}\\
V_{s,\ell}(-1)&=&\ell(\ell+1)+\frac{(1-s)(1+2s)}{2}\left(1-\frac{4x_0^2\mu^3}{\sqrt{\pi}}e^{-\mu^2 x_0^2}\right).
\label{fv1A}
\end{eqnarray}

\begin{table}
\caption{Classification of the points $y=\pm 1$ for the relevant functions defined by  (\ref{S2ononeA})-(\ref{S0ononeA}), and (\ref{fyA})-(\ref{fv1A}). The abbreviation $z$ ord $n$ stands for zero of order $n$. Notice that $V_{s,\ell}(-1)$ has been reported in \eqref{fv1A}.}
\begin{center}
\begin{tabular}{ | c | c | c | c | c | c | c | c }
\hline
$y$  & $f(y)$         & $V_{s,\ell}(y)$  & $\mathfrak{S}_2(y)$ & $\mathfrak{S}_1(y)$  & $\mathfrak{S}_0(y)$\\ \hline
$-1$ & z \mbox{ord} 1 & $V_{s,\ell}(-1)$ & z \mbox{ord} 1      & $2\dot{f}(-1)$       & $x_0^2\Omega^2+\frac{i\Omega}{2}(2x_0+1)\dot{f}(-1)-V_{s,\ell}(-1)$ \\ \hline
$+1$ & $+1$           & z \mbox{ord} 2   & z \mbox{ord} 2      & $4i x_0\Omega$       & $\frac{3}{4}\Omega^2-\ell(\ell+1)$\\ \hline
\end{tabular}
\label{tableEinsnoneA}
\end{center}
\end{table}

Notice that the coefficients of the differential equation (\ref{ODEynoneA}) do not share common zeros or poles at $y = \pm 1$ (see Table~\ref{tableEinsnoneA}). This means that we can directly apply the spectral method to\eqref{ODEynoneA}. To this purpose, it is convenient to rewrite $\mathfrak{S}_1(y)$ and $\mathfrak{S}_0(y)$ as follows
\begin{equation}\label{S210hononeA}
\mathfrak{S}_1(y) = i\Omega \widehat{N}_1(y)+\widehat{N}_0(y), \qquad
\mathfrak{S}_0(y) = \Omega^2 \widehat{C}_2(y)+i\Omega \widehat{C}_1(y)+\widehat{C}_0(y)
\end{equation}
with
\begin{eqnarray}
    \widehat{N}_1(y) &=&(4x_0+1-y)f(y), \quad
    \widehat{N}_0(y) = -2(1-y)f(y)+\frac{(1-y)^2}{2}\dot{f}(y),\label{N0A}\\
    \widehat{C}_2(y) &=&\frac{4x_0^2(1-f(y))}{(1-y)^2}-\left(\frac{2x_0}{1-y}+\frac{1}{4}\right)f(y),\label{C2A}\\
    \widehat{C}_1(y) &=&\frac{1}{2}\left[\left(2x_0+\frac{1-y}{2}\right)\dot{f}(y)-f(y)\right],\label{C1A}\\
    \widehat{C}_0(y) &=&-\frac{4V_{s,\ell}(y)}{(1-y)^2}.\label{C0A}
\end{eqnarray}
It can be easily verified with Maple that
\begin{eqnarray}
    &&\lim_{y\to 1^{-}}\mathfrak{S}_2(y)=0=\lim_{y\to -1^{+}}\mathfrak{S}_2(y),\\
    &&\lim_{y\to 1^{-}}\mathfrak{S}_1(y)=4ix_0\Omega,\quad
    \lim_{y\to -1^{+}}M_1(y)=\Lambda_0,\\
    &&\lim_{y\to 1^{-}}\mathfrak{S}_0(y)=\frac{3}{4}\Omega^2-\ell(\ell+1),\quad
     \lim_{y\to -1^{+}}\mathfrak{S}_0(y)=B_2\Omega^2+i\Omega B_1+B_0,
\end{eqnarray}
where
\begin{eqnarray}
\Lambda_0&=&1-\frac{4x_0^2\mu^3}{\sqrt{\pi}}e^{-\mu^2 x_0^2},\\
B_2 &=& x_0^2,\quad
B_1 =\frac{1+2x_0}{4}\left(1-\frac{4\mu^3 x_0^2}{\sqrt{\pi}}e^{-\mu^2 x_0^2}\right),\\
B_0 &=&\frac{(2s+1)(s-1)}{2}\left(1-\frac{4\mu^3 x_0^2}{\sqrt{\pi}}e^{-\mu^2 x_0^2}\right)-\ell(\ell+1).\label{B0A}
\end{eqnarray}
As a final step prior to implementing the spectral method, we rewrite the differential equation (\ref{ODEynoneA}) in the following form
\begin{equation}\label{TSCHNCW}
  \widehat{L}_0\left[\Phi_{\Omega\ell s}, \dot{\Phi}_{\Omega\ell s}, \ddot{\Phi}_{\Omega\ell s}\right] +  i\widehat{L}_1\left[\Phi_{\Omega\ell s}, \dot{\Phi}_{\Omega\ell s}, \ddot{\Phi}_{\Omega\ell s}\right]\Omega +  \widehat{L}_2\left[\Phi_{\Omega\ell s}, \dot{\Phi}_{\Omega\ell s}, \ddot{\Phi}_{\Omega\ell s}\right]\Omega^2 = 0
\end{equation}
with
\begin{eqnarray}
  \widehat{L}_0\left[\Phi_{\Omega\ell s}, \dot{\Phi}_{\Omega\ell s}, \ddot{\Phi}_{\Omega\ell s}\right] &=& \widehat{L}_{00}(y)\Phi_{\Omega\ell s} + \widehat{L}_{01}(y)\dot{\Phi}_{\Omega\ell s} + \widehat{L}_{02}(y)\ddot{\Phi}_{\Omega\ell s},\label{L0noneM}\\
  \widehat{L}_1\left[\Phi_{\Omega\ell s}, \dot{\Phi}_{\Omega\ell s}, \ddot{\Phi}_{\Omega\ell s}\right] &=& \widehat{L}_{10}(y)\Phi_{\Omega\ell s} + \widehat{L}_{11}(y)\dot{\Phi}_{\Omega\ell s} + \widehat{L}_{12}(y)\ddot{\Phi}_{\Omega\ell s}, \label{L1noneM}\\
  \widehat{L}_2\left[\Phi_{\Omega\ell s}, \dot{\Phi}_{\Omega\ell s}, \ddot{\Phi}_{\Omega\ell s}\right] &=& \widehat{L}_{20}(y)\Phi_{\Omega\ell s} + \widehat{L}_{21}(y)\dot{\Phi}_{\Omega\ell s} + \widehat{L}_{22}(y)\ddot{\Phi}_{\Omega\ell s}.\label{L2noneM}
\end{eqnarray}
Furthermore, Table~\ref{tableZweinone} provides a summary of the $\widehat{L}_{ij}$ terms from (\ref{L0noneM})-(\ref{L2noneM}) along with their limiting values at \(y = \pm 1\).

\begin{table}
\caption{We present the definitions of the coefficients \(\widehat{L}_{ij}\) and their behaviors at the endpoints of the interval \(-1 \leq y \leq 1\). The symbols used in the table are defined in \eqref{S210hononeA} and \eqref{N0A}-\eqref{C0A}.}
\begin{center}
\begin{tabular}{ | c | c | c | c | c | c | c | c }
\hline
$(i,j)$  & $\displaystyle{\lim_{y\to -1^+}}\widehat{L}_{ij}$  & $\widehat{L}_{ij}$ & $\displaystyle{\lim_{y\to 1^-}}\widehat{L}_{ij}$  \\ \hline
$(0,0)$ &  $B_0$          & $\widehat{C}_0$                  & $-\ell(\ell+1)$\\ \hline
$(0,1)$ &  $\Lambda_0$    & $\widehat{N}_0$                  & $0$\\ \hline
$(0,2)$ &  $0$            & $\widehat{\mathfrak{S}}_2$                  & $0$\\ \hline 
$(1,0)$ &  $B_1$          & $\widehat{C}_1$                  & $0$\\ \hline 
$(1,1)$ &  $0$    & $\widehat{N}_1$                  & $4x_0$\\ \hline 
$(1,2)$ &  $0$            & $0$                    & $0$\\ \hline 
$(2,0)$ &  $x_0^2$          & $\widehat{C}_2$                  & $\frac{3}{4}$\\ \hline
$(2,1)$ &  $0$            & $0$                    & $0$\\ \hline
$(2,2)$ &  $0$            & $0$                    & $0$\\ \hline
\end{tabular}
\label{tableZweinone}
\end{center}
\end{table} 

\subsection{Case B}
\label{subBNC}

In this scenario, we select the Frobenius solution corresponding to the exponent $\rho = 1/2$, which is equivalent to the requirement $\psi_{\Omega\ell s}(1) = 0$. Using the following ansatz, which encodes the correct behaviour at the throat and asymptotically at space-like infinity ($z \to +\infty$)
\begin{equation}  
\psi_{\Omega\ell\epsilon}(z) = \sqrt{1 - \frac{1}{z}} \, z^{\frac{i\Omega}{2}} e^{ix_0\Omega(z-1)} U_{\Omega\ell s}(z),  \end{equation}  
and imposing the condition that $U_{\Omega\ell s}(z)$ remains regular as $z \to 1^+$ and $z \to +\infty$, equation \eqref{ODEor} becomes
\begin{equation}\label{ODEznoneB}
    \widehat{\mathfrak{P}}_2(z) \Phi''_{\Omega \ell s}(z) + \widehat{\mathfrak{P}}_1(z) \Phi'_{\Omega \ell s}(z) + \widehat{\mathfrak{P}}_0(z) \Phi_{\Omega \ell s}(z) = 0
\end{equation}
with
\begin{eqnarray}
    \widehat{\mathfrak{P}}_2(z)&=&z^2(z-1)^2 f(z),\\
    \widehat{\mathfrak{P}}_1(z)&=&\frac{z(z-1)}{2}\left[z(z-1)f^{'}(z)+2f(z)+2i\Omega(z-1)(2x_0 z+1)f(z)\right],\\
    \widehat{\mathfrak{P}}_0(z)&=&-\frac{(z-1)^2}{4}\left[(2x_0 z+1)^2 f(z)-4x_0^2 z^2\right]\Omega^2\nonumber\\
    &&+\frac{i\Omega}{4}(z-1)\left[z(z-1)(2x_0 z+1)f^{'}(z)+2f(z)(2x_0 z-z+2)\right]\nonumber\\
    &&+\frac{z(z-1)}{4}f^{'}(z)+\frac{3}{4}f(z)-zf(z)-z^2(z-1)^2 V_{s,\ell}(z).\\
\end{eqnarray}
As in the previous subsection, we apply the transformation $z = 2/(1-y)$, which maps the point at infinity and the throat to $y = 1$ and $y = -1$, respectively. In the following, a dot stands for differentiation with respect to the new variable $y$. As a result, equation \eqref{ODEznoneB} becomes
\begin{equation}\label{ODEynoneB}
    \widehat{\mathfrak{S}}_2(y)\ddot{\Phi}_{\Omega\ell s}(y) + \widehat{\mathfrak{S}}_1(y)\dot{\Phi}_{\Omega\ell s}(y) + \widehat{\mathfrak{S}}_0(y)\Phi_{\Omega\ell s}(y) = 0,
\end{equation}
where
\begin{eqnarray}
  \widehat{\mathfrak{S}}_2(y) &=&(1+y)^2 f(y), \label{S2ononeB} \\
  \widehat{\mathfrak{S}}_1(y) &=&\frac{1+y}{1-y}\left[\frac{1-y^2}{2}f^{'}(y)-(3y+1)f(y)\right]+i\Omega\left(\frac{1+y}{1-y}\right)^2(4x_0+1-y)f(y),\label{S1ononeB}\\
  \widehat{\mathfrak{S}}_0(y) &=&-\frac{1}{4}\left(\frac{1+y}{1-y}\right)^2\left[\left(\frac{4x_0}{1-y}+1\right)^2 f(y)-\frac{16 x_0^2}{(1-y)^2}\right]\Omega^2\nonumber\\
  &&+\frac{i\Omega}{4}\frac{1+y}{(1-y)^2}\left[(1+y)(1+4x_0-y)\dot{f}(y)+4(2x_0-y)f(y)\right]\nonumber\\
  &&+\frac{1+y}{4}\dot{f}(y)-\frac{5+3y}{4(1-y)}f(y)-\frac{4(1+y)^2}{(1-y)^4}V_{s,\ell}(y). \label{S0ononeB}
\end{eqnarray}
In addition, we also require that $\Phi_{\Omega\ell s}(y)$ remains regular at $y=\pm 1$. Notice that $f(y)$ and $V_{s,\ell}(y)$ have been already given in \eqref{fyA} and \eqref{fv1A}.

\begin{table}
\caption{Classification of the points $y=\pm 1$ for the relevant functions defined by  (\ref{S2ononeB})-(\ref{S0ononeB}). The abbreviations $z$ ord $n$ and $p$ ord $m$ stand for zero of order $n$ and pole of order $m$, respectively.}
\begin{center}
\begin{tabular}{ | c | c | c | c | c | c }
\hline
$y$   & $\widehat{\mathfrak{S}}_2(y)$ & $\widehat{\mathfrak{S}}_1(y)$  & $\widehat{\mathfrak{S}}_0(y)$\\ \hline
$-1$  & z \mbox{ord} 3      & z \mbox{ord} 2       & z \mbox{ord} 2 \\ \hline
$+1$  & 4                   & p \mbox{ord} 2       & p \mbox{ord} 2\\ \hline
\end{tabular}
\label{tableEinsnoneB}
\end{center}
\end{table}

We observe that the coefficients of the differential equation (\ref{ODEynoneB}) share a common zero of order two at $y=-1$ while both $\widehat{\mathfrak{S}}_1$ and $\widehat{\mathfrak{S}}_2$ exhibit a pole of order two at $y=1$.  Hence, in order to apply the spectral method, we need to multiply (\ref{ODEynoneB}) by $(1-y)^2/(1+y)^2$. As a result, we end up with the following differential equation
\begin{equation}\label{ODEhynoneB}
    \mathfrak{M}_2(y)\ddot{\Phi}_{\Omega\ell\epsilon}(y) + \mathfrak{M}_1(y)\dot{\Phi}_{\Omega\ell\epsilon}(y) + \mathfrak{M}_0(y)\Phi_{\Omega\ell\epsilon}(y) = 0,
\end{equation}
where
\begin{equation}\label{S210hononeB}
  \mathfrak{M}_2(y) = (1-y)^2f(y), \qquad
  \mathfrak{M}_1(y) = i\Omega \mathfrak{N}_1(y)+\mathfrak{N}_0(y), \qquad
  \mathfrak{M}_0(y) = \Omega^2 \mathfrak{C}_2(y)+i\Omega \mathfrak{C}_1(y)+\mathfrak{C}_0(y)
\end{equation}
with
\begin{eqnarray}
    \mathfrak{N}_1(y) &=&(4x_0+1-y)f(y), \quad
    \mathfrak{N}_0(y) = \frac{1-y}{1+y}\left[\frac{1-y^2}{2}\dot{f}(y)-(3y+1)f(y)\right],\label{N0B}\\
    \mathfrak{C}_2(y) &=& -\frac{1}{4}\left(\frac{4x_0}{1-y}+1\right)^2 f(y)+\frac{4 x_0^2}{(1-y)^2},\label{C2B}\\
    \mathfrak{C}_1(y) &=&\frac{1}{4(1+y)}\left[(1+y)(1+4x_0-y)\dot{f}(y)+4(2x_0-y)f(y)\right],\label{C1B}\\
    \mathfrak{C}_0(y) &=&\left(\frac{1-y}{1+y}\right)^2\left[\frac{1+y}{4}\dot{f}(y)-\frac{5+3y}{4(1-y)}f(y)-\frac{4(1+y)^2}{(1-y)^4}V_{s,\ell}(y)\right].\label{C0B}
\end{eqnarray}
It can be easily verified with Maple that
\begin{eqnarray}
    &&\lim_{y\to 1^{-}}\mathfrak{M}_2(y)=0=\lim_{y\to -1^{+}}\mathfrak{M}_2(y),\\
    &&\lim_{y\to 1^{-}}\mathfrak{M}_1(y)=4ix_0\Omega,\quad
    \lim_{y\to -1^{+}}\mathfrak{M}_1(y)=\widehat{\Lambda}_0,\\
    &&\lim_{y\to 1^{-}}\mathfrak{M}_0(y)=\frac{3}{4}\Omega^2-\ell(\ell+1)+ix_0\Omega,\quad
     \lim_{y\to -1^{+}}\mathfrak{M}_0(y)=\widehat{B}_2\Omega^2+i\Omega\widehat{B}_1+\widehat{B}_0,
\end{eqnarray}
where
\begin{eqnarray}
\widehat{\Lambda}_0&=&3-\frac{12x_0^2\mu^3}{\sqrt{\pi}}e^{-\mu^2 x_0^2},\\
\widehat{B}_2 &=& x_0^2,\quad
\widehat{B}_1 =\frac{3(1+2x_0)}{4}\left(1-\frac{4\mu^3 x_0^2}{\sqrt{\pi}}e^{-\mu^2 x_0^2}\right),\\
\widehat{B}_0 &=&\frac{s(2s-1)}{2}\left(1-\frac{4\mu^3 x_0^2}{\sqrt{\pi}}e^{-\mu^2 x_0^2}\right)+\frac{2x_0^4\mu^5+10x_0^2\mu^3-3\sqrt{\pi}e^{\mu^2 x_0^2}}{2\sqrt{\pi}e^{\mu^2 x_0^2}}-\ell(\ell+1).\label{B0AM}
\end{eqnarray}
As a final step prior to implementing the spectral method, we rewrite the differential equation (\ref{ODEynoneB}) in the following form
\begin{equation}\label{TSCHB}
  \widehat{\widehat{L}}_0\left[\Phi_{\Omega\ell s}, \dot{\Phi}_{\Omega\ell s}, \ddot{\Phi}_{\Omega\ell s}\right] +  i\widehat{\widehat{L}}_1\left[\Phi_{\Omega\ell s}, \dot{\Phi}_{\Omega\ell s}, \ddot{\Phi}_{\Omega\ell s}\right]\Omega +  \widehat{\widehat{L}}_2\left[\Phi_{\Omega\ell s}, \dot{\Phi}_{\Omega\ell s}, \ddot{\Phi}_{\Omega\ell s}\right]\Omega^2 = 0
\end{equation}
with
\begin{eqnarray}
  \widehat{\widehat{L}}_0\left[\Phi_{\Omega\ell s}, \dot{\Phi}_{\Omega\ell s}, \ddot{\Phi}_{\Omega\ell s}\right] &=& \widehat{\widehat{L}}_{00}(y)\Phi_{\Omega\ell s} + \widehat{\widehat{L}}_{01}(y)\dot{\Phi}_{\Omega\ell s} + \widehat{\widehat{L}}_{02}(y)\ddot{\Phi}_{\Omega\ell s},\label{L0noneB}\\
  \widehat{\widehat{L}}_1\left[\Phi_{\Omega\ell s}, \dot{\Phi}_{\Omega\ell s}, \ddot{\Phi}_{\Omega\ell s}\right] &=& \widehat{\widehat{L}}_{10}(y)\Phi_{\Omega\ell s} + \widehat{\widehat{L}}_{11}(y)\dot{\Phi}_{\Omega\ell s} + \widehat{\widehat{L}}_{12}(y)\ddot{\Phi}_{\Omega\ell s}, \label{L1noneB}\\
  \widehat{\widehat{L}}_2\left[\Phi_{\Omega\ell s}, \dot{\Phi}_{\Omega\ell s}, \ddot{\Phi}_{\Omega\ell s}\right] &=& \widehat{\widehat{L}}_{20}(y)\Phi_{\Omega\ell s} + \widehat{\widehat{L}}_{21}(y)\dot{\Phi}_{\Omega\ell s} + \widehat{\widehat{L}}_{22}(y)\ddot{\Phi}_{\Omega\ell s}.\label{L2noneB}
\end{eqnarray}
Furthermore, Table~\ref{tableZweinoneB} provides a summary of the $\widehat{\widehat{L}}_{ij}$ terms from (\ref{L0noneB})-(\ref{L2noneB}) along with their limiting values at \(y = \pm 1\).

\begin{table}
\caption{We present the definitions of the coefficients $\widehat{\widehat{L}}_{ij}$ and their behaviours at the endpoints of the interval $-1 \leq y \leq 1$. The symbols used in the table are defined in \eqref{S210hononeB} and \eqref{N0B}-\eqref{C0B}.}
\begin{center}
\begin{tabular}{ | c | c | c | c | c | c | c | c }
\hline
$(i,j)$  & $\displaystyle{\lim_{y\to -1^+}}\widehat{\widehat{L}}_{ij}$  & $\widehat{\widehat{L}}_{ij}$ & $\displaystyle{\lim_{y\to 1^-}}\widehat{\widehat{L}}_{ij}$  \\ \hline
$(0,0)$ &  $\widehat{B}_0$          & $\mathfrak{C}_0$                  & $-\ell(\ell+1)$\\ \hline
$(0,1)$ &  $\widehat{\Lambda}_0$    & $\mathfrak{N}_0$                  & $0$\\ \hline
$(0,2)$ &  $0$                      & $\mathfrak{M}_2$                  & $0$\\ \hline 
$(1,0)$ &  $\widehat{B}_1$          & $\mathfrak{C}_1$                  & $x_0$\\ \hline 
$(1,1)$ &  $0$    & $\mathfrak{N}_1$                  & $4x_0$\\ \hline 
$(1,2)$ &  $0$            & $0$                    & $0$\\ \hline 
$(2,0)$ &  $x_0^2$          & $\mathfrak{C}_2$                  & $\frac{3}{4}$\\ \hline
$(2,1)$ &  $0$            & $0$                    & $0$\\ \hline
$(2,2)$ &  $0$            & $0$                    & $0$\\ \hline
\end{tabular}
\label{tableZweinoneB}
\end{center}
\end{table} 
We conclude this section with the following theorem, which provides an explanation for the appearance of both positive and negative signs in the real part of the QNMs.
\begin{theorem}
Consider the polynomial eigenvalue problem
\begin{equation}\label{QNMeq}
\mathfrak{L}_0\left[\mathfrak{U}(y), \dot{\mathfrak{U}}(y), \ddot{\mathfrak{U}}(y)\right]
+i\mathfrak{L}_1\left[\mathfrak{U}(y), \dot{\mathfrak{U}}(y), \ddot{\mathfrak{U}}(y)\right]\Omega
+\mathfrak{L}_2\left[\mathfrak{U}(y), \dot{\mathfrak{U}}(y), \ddot{\mathfrak{U}}(y)\right]\Omega^2=0
\end{equation}
on the interval $[-1,1]$ where $\mathfrak{U}(y)$ is complex-valued, $\Omega=\Omega_R+i\Omega_I$ is a complex eigenvalue (QNM), and the dot denotes differentiation with respect to the variable $y\in[-1,1]$. Furthermore, the operators $\mathfrak{L}_i$ with $i\in\{0,1,2\}$ are linear and have the form
\begin{equation}
\mathfrak{L}_i\left[\mathfrak{U}(y), \dot{\mathfrak{U}}(y), \ddot{\mathfrak{U}}(y)\right]=
\mathfrak{L}_{i0}\mathfrak{U}(y)+\mathfrak{L}_{i1}\dot{\mathfrak{U}}(y)+\mathfrak{L}_{i2}\ddot{\mathfrak{U}}(y).
\end{equation}
If $\Omega$ is an eigenvalue of the above problem with the corresponding eigenfunction $\mathfrak{U}(y)$, then $\Omega^{*}=-\Omega_R + i\Omega_I$ is also an eigenvalue of the same problem, associated with the eigenfunction $\overline{\mathfrak{U}(y)}$, which denotes the complex conjugate of $\mathfrak{U}(y)$.
\end{theorem}
\begin{proof}
The result can be readily obtained by applying the complex conjugate operation to the equation \eqref{QNMeq} and noting that all coefficients of the involved differential operators are real-valued.\hfill \qedsymbol
\end{proof}

\section{Numerical method}

To solve the differential eigenvalue problems \eqref{TSCH}, \eqref{TSCHNCW}, and \eqref{TSCHB} in order to determine the QNMs and their corresponding frequencies $\Omega$, we discretize the differential operators as defined in equations (\ref{L0none})–(\ref{L2none}), (\ref{L0noneM})–(\ref{L2noneM}), and  (\ref{L0noneB})–(\ref{L2noneB}). Since the problem is posed on the finite interval $[-1, 1]$ and only requires that the regular part of the QNM eigenfunction be regular at $y = \pm 1$, a Chebyshev-type spectral method \cite{Trefethen2000, Boyd2000} is a natural choice. Specifically, the function $\Phi_{\Omega \ell s}(y)$ is expanded in a truncated Chebyshev series
\begin{equation}
\Phi_{\Omega \ell s}(y) = \sum_{k=0}^{N} a_k T_k(y),
\end{equation}
where $N \in \mathbb{N}$ is a numerical parameter, $\{a_k\}_{k=0}^{N} \subseteq \mathbb{R}$ are coefficients, and $\{T_k(y)\}_{k=0}^{N}$ are the Chebyshev polynomials of the first kind, defined by $T_k(y) = \cos(k \arccos y)$ for $y \in [-1, 1]$. Substituting this expansion into equations \eqref{TSCH}, \eqref{TSCHNCW}, and \eqref{TSCHB} transforms the problem into an eigenvalue problem with polynomial coefficients. To reformulate it into a numerical framework, the collocation method \cite{Boyd2000} is employed. Instead of requiring the polynomial function in $y$ to vanish identically, this approach enforces vanishing at $N+1$ selected collocation points, which match the number of unknown coefficients $\{a_k\}_{k=0}^{N}$. For these collocation points, we use the Chebyshev roots grid \cite{Fox1968}  
\begin{equation}
y_k = -\cos\left(\frac{(2k+1)\pi}{2(N+1)}\right), \quad k \in \{0, 1, \ldots, N\}.
\end{equation}
An alternative option, also implemented in our codes, is the Chebyshev extrema grid
\begin{equation}
y_k = -\cos\left(\frac{k\pi}{N}\right), \quad k \in \{0, 1, \ldots, N\}.
\end{equation}  
Our numerical computations used the roots grid, though the theoretical performance of both choices is known to be comparable \cite{Fox1968, Boyd2000}. Applying the collocation method results in a classical matrix-based quadratic eigenvalue problem \cite{Tisseur2001}  
\begin{equation}\label{eq:eig}
(M_0 + iM_1\Omega + M_2\Omega^2)\mathbf{a} = \mathbf{0},
\end{equation}
where $M_j$ with $j\in\{0, 1, 2\}$ are square real matrices of size $(N+1) \times (N+1)$ that represent the spectral discretizations of $L_j[\cdot]$, $\widetilde{L}_j[\cdot]$, $\widehat{L}_{j}[\cdot]$, or $\widehat{\widehat{L}}_{j}[\cdot]$. This eigenvalue problem is solved using the \texttt{polyeig} function in \textsc{Matlab}, yielding $2(N+1)$ potential eigenvalues for $\Omega$. To identify the physically meaningful QNM frequencies, root plots for different values of $N$ (e.g., $N = 100, 150, 200$) are overlapped, and consistent roots that remain stable across these values are selected.  To minimize numerical rounding errors and floating-point inaccuracies, all computations were performed using multiple precision arithmetic in \textsc{Maple}, interfaced with \textsc{Matlab} via the \textsc{Advanpix} toolbox \cite{mct2015}. The reported numerical results were calculated with a precision of 200 decimal digits.

\section{Numerical results}\label{Numresults}

In this section, we present our numerical findings for the QNMs of the classical Schwarzschild and the noncommutative geometry-inspired Schwarzschild wormholes, computed using the spectral method described earlier, with a focus on understanding the effects of noncommutativity, and the rescaled mass parameter on the ringdown spectrum. Additionally, we compare these results with the QNMs of the classical Schwarzschild wormhole to understand the degree of influence of the noncommutative parameter in the QNM spectrum.

Tables~\ref{table:1}-\ref{table:1aaextNCSW} provide several significant observations regarding the QNMs of noncommutative geometry-inspired wormholes across scalar, electromagnetic, and vector-type gravitational perturbations for large values of the mass parameter and in the nearly extremal regime. First of all, for both Schwarzschild wormholes and Morris-Thorne (MT) wormholes, significant differences emerge between the QNM spectra under Case I and Case II boundary conditions.  Case II tends to yield higher imaginary components $\Omega_I$, indicating faster decay rates. More specifically, for scalar perturbations and $\ell = 0$, $n = 0$, the imaginary part of the fundamental mode for the Schwarzschild wormhole increases from \(-0.3834i\) (Case I) to \(-1.3226i\) (Case II). Similar trends are observed for higher modes. For instance, for $\ell = 2, n = 3$, $\Omega_I$ increases from $-5.4882i$ (Case I) to $-6.5100i$ (Case II). Analogous considerations can be drawn from Tables~\ref{table:1a} and ~\ref{table:1b}. Hence, case II boundary conditions applied to the classical Schwarzschild wormhole lead to faster damping, likely due to increased energy dissipation at the wormhole's throat. This behaviour signalizes the sensitivity of QNMs to throat-specific conditions. Moreover, the imaginary parts of the QNMs for the MT wormhole are consistently smaller than those of the classical Schwarzschild wormhole across all types of perturbations examined in this work (see Tables~\ref{table:1},~\ref{table:1a}, and ~\ref{table:1b}). This indicates that the ringdown of the QNM spectrum of the MT wormhole is characterized by slower damping, reflecting its less dissipative nature in comparison to the classical Schwarzschild wormhole.

It is interesting to observe that, for large mass parameters ($\mu = 10^3$), the noncommutative Schwarzschild wormhole transitions smoothly into the classical Schwarzschild wormhole. For example, for scalar perturbations, the QNM spectra for the aforementioned wormholes are nearly identical for large $\mu$. Indeed, for $\ell = 0, n = 0$ we have $\Omega_{SW} = 0.5338-0.3834i$ (Case I), and $\Omega_{NCSW} = 0.5338-0.3834i$ (Case A). Moreover, for $\ell = 1, n = 3$ we find $\Omega_{SW} = 0.6712-5.9832i$ (Case I), and $\Omega_{NCSW} = 0.6711-5.9829i$ (Case A). The near-identical QNM spectra for both wormholes in the limit of a large mass parameter (see Table~\ref{table:crossvalidation}) confirm the robustness and consistency of our numerical method. Last but not least, it is gratifying to observe how the noncommutative model converges to its classical limit when $\mu \gg 1$.

In the nearly extremal regime, for $\mu = 1.91$ and $\mu = 1.905$, noncommutative effects become more pronounced, as seen in Table~\ref{table:1extNCSW}, \ref{table:1aextNCSW}, and \ref{table:1aaextNCSW}. More precisely, the noncommutative Schwarzschild wormhole exhibits distinct QNM spectra with respect to the classical Schwarzschild and the MT wormholes. Higher damping rates characterize both cases A and B for fixed $\ell$ and increasing $n$. Also, in this case, such behaviour reflects the influence of the throat's nearly extreme geometry. Hence, we conclude that noncommutative geometry introduces distinct spectral signatures that become more pronounced near extremality, providing a potential observational window for testing such models. Furthermore, the QNM frequencies for scalar perturbations tend to display a slower decay (smaller $|\Omega_I|$) compared to electromagnetic and gravitational perturbations. 

Last but not least, our analysis reveals that no overdamped modes were detected for the noncommutative geometry-inspired wormhole, indicating stability under the considered perturbations. Interestingly, this stability extends to the classical limit as well, where the QNMs of the noncommutative wormhole converge to those of a classical Schwarzschild wormhole with $g_{tt} = -1$, rather than the Schwarzschild form $g_{tt} = -(1 - 2M/r)$. It is important to emphasize that this distinction sets our work apart from the well-known results of \cite{Fuller1962PR}, who demonstrated that the Einstein-Rosen bridge, characterized by $g_{tt} = -(1 - 2M/r)$, is dynamically unstable. In our case, the difference in the time component of the metric ensures that these classical instability results do not directly apply. Consequently, our numerical analysis confirms that both the noncommutative wormhole and its classical limit exhibit stable oscillatory QNMs without collapse, offering new insights into the stability properties of wormholes within the framework of noncommutative geometry.

\begin{table}
\centering
\caption{The table below presents the QNMs for scalar perturbations ($s = 0$) of the Schwarzschild wormhole, alongside those of the Morris-Thorne wormhole, as computed using the Spectral Method by \cite{Batic2024CGG}. The last two columns display values obtained via the Spectral Method, employing 200 polynomials to achieve a precision of 200 digits. "Case I" and "Case II" correspond to the sets of QNM boundary conditions described in subsections~\ref{subA} and \ref{subB}, respectively. Here, $\Omega = 2M\omega$ represents the dimensionless frequency. The abbreviations 'SW' and 'MT' stand for 'Schwarzschild Wormhole' and 'Morris-Thorne Wormhole', respectively.}
\label{table:1}
\vspace*{1em}
\begin{tabular}{||c@{\hspace{0.5em}}|@{\hspace{0.5em}}c@{\hspace{0.5em}}|@{\hspace{0.5em}}c@{\hspace{0.5em}}|@{\hspace{0.5em}}c@{\hspace{0.5em}}|@{\hspace{0.5em}}c@{\hspace{0.5em}}||}
\Xhline{2pt}
\rule{0pt}{1.0\normalbaselineskip}
$\ell$ & $n$ & $\Omega_{SW}$ \text{(Case I)} & $\Omega_{SW}$ \text{(Case II)} & $\Omega_{MT}$ \cite{Batic2024CGG}  \\ [0.25em]
\Xhline{2pt}\hline
\rule{0pt}{1.15\normalbaselineskip}
$0$ & $0$ & $0.5338-0.3834i$ & $0.3719-1.3226i$ & $0.6814 - 0.6178i$ \\
    & $1$ & $0.3003-2.3533i$ & $0.2695-3.3777i$ & $0.4672 - 2.1765i$ \\ [0.5ex]
\hline
\rule{0pt}{1.15\normalbaselineskip}
$1$ & $0$ & $1.5106-0.3575i$ & $1.3927-1.1046i$ & $1.5727 - 0.5297i$ \\
    & $1$ & $1.1938-1.9457i$ & $0.9944-2.8990i$ & $1.2558 - 1.7025i$ \\
    & $2$ & $0.8460-3.9147i$ & $0.7436-4.9483i$ & $0.8368 - 3.2361i$ \\ 
    & $3$ & $0.6712-5.9832i$ & $0.6175-7.0149i$ & $0.6334 - 4.9344i$ \\ [0.5ex]
\hline
\rule{0pt}{1.15\normalbaselineskip}
$2$ & $0$ & $2.5063-0.3550i$ & $2.4327-1.0764i$ & $2.5467 - 0.5127i$ \\
    & $1$ & $2.2915-1.8328i$ & $2.0991-2.6477i$ & $2.3450 - 1.5725i$ \\
    & $2$ & $1.8852-3.5353i$ & $1.6828-4.4895i$ & $1.9478 - 2.7604i$ \\
    & $3$ & $1.5113-5.4882i$ & $1.3730-6.5100i$ & $1.4533 - 4.2052i$ \\ [0.5ex]
\hline
\rule{0pt}{1.15\normalbaselineskip}
$3$ & $0$ & $3.5045-0.3543i$ & $3.4514-1.0687i$ & $3.5343 - 0.5069i$ \\
    & $1$ & $3.3472-1.8007i$ & $3.1966-2.5628i$ & $3.3901 - 1.5372i$ \\
    & $2$ & $3.0089-3.3671i$ & $2.7986-4.2224i$ & $3.0999 - 2.6232i$ \\
    & $3$ & $2.5837-5.1305i$ & $2.3805-6.0850i$ & $2.6722 - 3.8235i$ \\ [0.5ex]
\hline
\rule{0pt}{1.15\normalbaselineskip}
$4$ & $0$ & $4.5035-0.3540i$ & $4.4620-1.0655i$ & $4.5271 - 0.5043i$ \\
    & $1$ & $4.3800-1.7876i$ & $4.2594-2.5277i$ & $4.4151 - 1.5227i$ \\
    & $2$ & $4.1038-3.2933i$ & $3.9191-4.0920i$ & $4.1899 - 2.5722i$ \\
    & $3$ & $3.7139-4.9293i$ & $3.4988-5.8083i$ & $3.8511 - 3.6818i$ \\ [1ex] 
 \Xhline{2pt}
 \end{tabular}
\end{table}

\begin{table}
\centering
\caption{The table below presents the QNMs for electromagnetic perturbations ($s = 1$) of the Schwarzschild wormhole, alongside those of the Morris-Thorne wormhole, as computed using the Spectral Method by \cite{Batic2024CGG}. The last two columns display values obtained via the Spectral Method, employing 200 polynomials to achieve a precision of $200$ digits. "Case I" and "Case II" correspond to the sets of QNM boundary conditions described in subsections~\ref{subA} and \ref{subB}, respectively. Here, $\Omega = 2M\omega$ represents the dimensionless frequency. The abbreviations 'SW' and 'MT' stand for 'Schwarzschild Wormhole' and 'Morris-Thorne Wormhole', respectively.}
\label{table:1a}
\vspace*{1em}
\begin{tabular}{||c@{\hspace{0.5em}}|@{\hspace{0.5em}}c@{\hspace{0.5em}}|@{\hspace{0.5em}}c@{\hspace{0.5em}}|@{\hspace{0.5em}}c@{\hspace{0.5em}}|@{\hspace{0.5em}}c@{\hspace{0.5em}}||}
\Xhline{2pt}
\rule{0pt}{1.0\normalbaselineskip}
$\ell$ & $n$ & $\Omega_{SW}$ \text{(Case I)} & $\Omega_{SW}$ \text{(Case II)} & $\Omega_{MT}$ \cite{Batic2024CGG}  \\ [0.25em]
\Xhline{2pt}\hline
\rule{0pt}{1.15\normalbaselineskip}
$1$ & $0$ & $1.3412-0.3376i$ & $1.2120-1.0462i$ & $1.2665 - 0.4440i$ \\
    & $1$ & $0.9907-1.8583i$ & $0.7695-2.7965i$ & $0.9407 - 1.4136i$ \\
    & $2$ & $0.6061-3.8014i$ & $0.4914-4.8241i$ & $0.3991 - 2.7252i$ \\ [0.5em]
\hline
\rule{0pt}{1.15\normalbaselineskip}
$2$ & $0$ & $2.4058-0.3479i$ & $2.3300-1.0552i$ & $2.3549 - 0.4777i$ \\
    & $1$ & $2.1842-1.7983i$ & $1.9849-2.6018i$ & $2.1592 - 1.4631i$ \\
    & $2$ & $1.7631-3.4810i$ & $1.5538-4.4293i$ & $1.7689 - 2.5602i$ \\
    & $3$ & $1.3770-5.4232i$ & $1.2345-6.4403i$ & $1.2591 - 3.8905i$ \\ [0.5em]
\hline
\rule{0pt}{1.15\normalbaselineskip}
$3$ & $0$ & $3.4329-0.3507i$ & $3.3790-1.0578i$ & $3.3949 - 0.4881i$ \\
    & $1$ & $3.2732-1.7828i$ & $3.1202-2.5382i$ & $3.2542 - 1.4797i$ \\
    & $2$ & $2.9291-3.3365i$ & $2.7150-4.1868i$ & $2.9706 - 2.5232i$ \\
    & $3$ & $2.4963-5.0911i$ & $2.2898-6.0425i$ & $2.5513 - 3.6734i$ \\ [0.5em]
\hline
\rule{0pt}{1.15\normalbaselineskip}
$4$ & $0$ & $4.4478-0.3518i$ & $4.4061-1.0589i$ & $4.4177 - 0.4926i$ \\
    & $1$ & $4.3233-1.7767i$ & $4.2016-2.5126i$ & $4.3078 - 1.4872i$ \\
    & $2$ & $4.0444-3.2743i$ & $3.8578-4.0692i$ & $4.0864 - 2.5118i$ \\
    & $3$ & $3.6504-4.9033i$ & $3.4329-5.7796i$ & $3.7534 - 3.5941i$ \\ [1ex]
 \Xhline{2pt}
 \end{tabular}
\end{table}

\begin{table}
\centering
\caption{The table below presents the QNMs for vector-type gravitational perturbations ($s = 2$) of the Schwarzschild wormhole, alongside those of the Morris-Thorne wormhole, as computed using the Spectral Method by \cite{Batic2024CGG}. The last two columns display values obtained via the Spectral Method, employing 200 polynomials to achieve a precision of 200 digits. "Case I" and "Case II" correspond to the sets of QNM boundary conditions described in subsections~\ref{subA} and \ref{subB}, respectively. Here, $\Omega = 2M\omega$ represents the dimensionless frequency. The abbreviations 'SW' and 'MT' stand for 'Schwarzschild Wormhole' and 'Morris-Thorne Wormhole', respectively.}
\label{table:1b}
\vspace*{1em}
\begin{tabular}{||c@{\hspace{0.5em}}|@{\hspace{0.5em}}c@{\hspace{0.5em}}|@{\hspace{0.5em}}c@{\hspace{0.5em}}|@{\hspace{0.5em}}c@{\hspace{0.5em}}|@{\hspace{0.5em}}c@{\hspace{0.5em}}||}
\Xhline{2pt}
\rule{0pt}{1.0\normalbaselineskip}
$\ell$ & $n$ & $\Omega_{SW}$ \text{(Case I)} & $\Omega_{SW}$ \text{(Case II)} & $\Omega_{MT}$ \cite{Batic2024CGG}  \\ [0.25em]
\Xhline{2pt}\hline
\rule{0pt}{1.15\normalbaselineskip}
$2$ & $0$ & $1.8360-0.2833i$ & $1.7773-0.8694i$ & $1.7377 - 0.3051i$ \\
    & $1$ & $1.6551-1.5123i$ & $1.4762-2.2450i$ & $1.7203 - 1.0396i$ \\
    & $2$ & $1.2744-3.0849i$ & $1.0934-4.0150i$ & $1.5249 - 2.0305i$ \\
    & $3$ & $0.9524-4.9973i$ & $0.8472-6.0029i$ & $1.1699 - 3.3340i$ \\
    & $4$ & $0.7680-7.0175i$ & $0.7070-8.0350i$ & $0.8819 - 4.8985i$ \\ [0.5em]
\hline
\rule{0pt}{1.15\normalbaselineskip}
$3$ & $0$ & $3.0513-0.3264i$ & $2.9968-0.9855i$ & $2.9524 - 0.4100i$ \\
    & $1$ & $2.8891-1.6638i$ & $2.7325-2.3752i$ & $2.8625 - 1.2591i$ \\
    & $2$ & $2.5363-3.1344i$ & $2.3167-3.9525i$ & $2.6606 - 2.1909i$ \\
    & $3$ & $2.0949-4.8319i$ & $1.8897-5.7637i$ & $2.3313 - 3.2596i$ \\
    & $4$ & $1.7113-6.7330i$ & $1.5611-7.7255i$ & $1.9109 - 4.5290i$ \\ [0.5em]
\hline
\rule{0pt}{1.15\normalbaselineskip}
$4$ & $0$ & $4.1589-0.3388i$ & $4.1160-1.0200i$ & $4.0763 - 0.4491i$ \\
    & $1$ & $4.0310-1.7122i$ & $3.9057-2.4230i$ & $3.9846 - 1.3592i$ \\
    & $2$ & $3.7436-3.1607i$ & $3.5507-3.9334i$ & $3.7971 - 2.3060i$ \\
    & $3$ & $3.3362-4.7478i$ & $3.1118-5.6073i$ & $3.5086 - 3.3190i$ \\
    & $4$ & $2.8898-6.5102i$ & $2.6804-7.4508i$ & $3.1230 - 4.4367i$ \\ [0.5em]
\hline
\rule{0pt}{1.15\normalbaselineskip}
$5$ & $0$ & $5.2240-0.3442i$ & $5.1891-1.0349i$ & $5.1548 - 0.4673i$ \\
    & $1$ & $5.1198-1.7329i$ & $5.0168-2.4430i$ & $5.0729 - 1.4086i$ \\
    & $2$ & $4.8821-3.1705i$ & $4.7182-3.9208i$ & $4.9075 - 2.3712i$ \\
    & $3$ & $4.5294-4.6993i$ & $4.3215-5.5102i$ & $4.6566 - 3.3722i$ \\
    & $4$ & $4.1017-6.3565i$ & $3.8783-7.2388i$ & $4.3203 - 4.4337i$ \\ [1ex]
 \Xhline{2pt}
 \end{tabular}
\end{table}

\begin{table}
\centering
\caption{This table presents the QNMs for scalar perturbations ($s = 0$) of the classic Schwarzschild wormhole and a noncommutative geometry-inspired Schwarzschild wormhole with a large mass parameter $\mu = 10^3$, where the throat is located at $x_0 = 1$. The quasinormal frequencies were computed using the Spectral Method, employing 200 polynomials with a precision of 200 digits. As expected, the noncommutative Schwarzschild wormhole transitions into the classic Schwarzschild wormhole in the limit $\mu \gg 1$.   Columns three and four correspond to "Case I" and "Case II," representing the QNM boundary conditions described in subsections~\ref{subA} and \ref{subB}, respectively. Similarly, "Case A" and "Case B" refer to the boundary conditions outlined in subsections~\ref{subANC} and \ref{subBNC}. Here, $\Omega = 2M\omega$ denotes the dimensionless frequency. The abbreviations "SW" and "NCSW" stand for "Schwarzschild Wormhole" and "noncommutative geometry-inspired Schwarzschild Wormhole," respectively.}
\label{table:crossvalidation}
\vspace*{1em}
\begin{tabular}
{||c@{\hspace{0.5em}}|@{\hspace{0.5em}}c@{\hspace{0.5em}}|@{\hspace{0.5em}}c@{\hspace{0.5em}}|@{\hspace{0.5em}}c@{\hspace{0.5em}}|@{\hspace{0.5em}}c@{\hspace{0.5em}}|@{\hspace{0.5em}}c@{\hspace{0.5em}}||}
\Xhline{2pt}
\rule{0pt}{1.0\normalbaselineskip}
$\ell$ & $n$ & $\Omega_{SW}$ \text{(Case I)} & $\Omega_{SW}$ \text{(Case II)} & $\Omega_{NCSW}$ \text{Case A} &  $\Omega_{NCSW}$ \text{Case B} \\ [0.25em]
\Xhline{2pt}\hline
\rule{0pt}{1.15\normalbaselineskip}
$0$ & $0$ & $0.5338-0.3834i$ & $0.3719-1.3226i$ & $0.5338-0.3834i$ & $0.3719-1.3226i$ \\
    & $1$ & $0.3003-2.3533i$ & $0.2695-3.3777i$ & $0.3005-2.3530i$ & $0.2679-3.3748i$ \\ [0.5em]
\hline
\rule{0pt}{1.15\normalbaselineskip}
$1$ & $0$ & $1.5106-0.3575i$ & $1.3927-1.1046i$ & $1.5106-0.3575i$ & $1.3927-1.1046i$ \\
    & $1$ & $1.1938-1.9457i$ & $0.9944-2.8990i$ & $1.1938-1.9457i$ & $0.9944-2.8990i$ \\
    & $2$ & $0.8460-3.9147i$ & $0.7436-4.9483i$ & $0.8460-3.9147i$ & $0.7436-4.9483i$ \\ 
    & $3$ & $0.6712-5.9832i$ & $0.6175-7.0149i$ & $0.6711-5.9829i$ & $0.6184-7.0158i$ \\ [0.5em]
\hline
\rule{0pt}{1.15\normalbaselineskip}
$2$ & $0$ & $2.5063-0.3550i$ & $2.4327-1.0764i$ & $2.5063-0.3550i$ & $2.4327-1.0764i$ \\
    & $1$ & $2.2915-1.8328i$ & $2.0991-2.6477i$ & $2.2915-1.8328i$ & $2.0991-2.6477i$ \\
    & $2$ & $1.8852-3.5353i$ & $1.6828-4.4895i$ & $1.8852-3.5353i$ & $1.6828-4.4895i$ \\
    & $3$ & $1.5113-5.4882i$ & $1.3730-6.5100i$ & $1.5113-5.4882i$ & $1.3730-6.5100i$ \\ [0.5em]
\hline
\rule{0pt}{1.15\normalbaselineskip}
$3$ & $0$ & $3.5045-0.3543i$ & $3.4514-1.0687i$ & $3.5045-0.3543i$ & $3.4514-1.0687i$ \\
    & $1$ & $3.3472-1.8007i$ & $3.1966-2.5628i$ & $3.3472-1.8007i$ & $3.1966-2.5628i$ \\
    & $2$ & $3.0089-3.3671i$ & $2.7986-4.2224i$ & $3.0089-3.3671i$ & $2.7986-4.2224i$ \\
    & $3$ & $2.5837-5.1305i$ & $2.3805-6.0850i$ & $2.5837-5.1305i$ & $2.3805-6.0850i$ \\ [0.5em]
\hline
\rule{0pt}{1.15\normalbaselineskip}
$4$ & $0$ & $4.5035-0.3540i$ & $4.4620-1.0655i$ & $4.5035-0.3540i$ & $4.4620-1.0655i$ \\
    & $1$ & $4.3800-1.7876i$ & $4.2594-2.5277i$ & $4.3800-1.7876i$ & $4.2594-2.5277i$ \\
    & $2$ & $4.1038-3.2933i$ & $3.9191-4.0920i$ & $4.1038-3.2933i$ & $3.9191-4.0920i$ \\
    & $3$ & $3.7139-4.9293i$ & $3.4988-5.8083i$ & $3.7139-4.9293i$ & $3.4988-5.8083i$ \\ [1ex] 
 \Xhline{2pt}
 \end{tabular}
\end{table}

\begin{table}
\centering
\caption{QNMs for scalar perturbations ($s = 0$) of the noncommutative geometry-inspired Schwarzschild wormhole in the nearly extremal case are presented for various values of the mass parameter $\mu$. These numerical values were computed using the spectral method with $200$ polynomials and an accuracy of $200$ digits. "Case A" and "Case B" refer to the boundary conditions introduced in subsections~\ref{subANC} and \ref{subBNC}. Here, $\Omega = 2M\omega$ denotes the dimensionless frequency. Moreover, the mass parameter values $\mu = 1.91$ and $\mu = 1.905$ correspond to throat locations at $x_0 = 0.83098$ and $x_0 = 0.80849$, respectively.}
\label{table:1extNCSW}
\vspace*{1em}
\begin{tabular}{||c@{\hspace{0.5em}}|@{\hspace{0.5em}}c@{\hspace{0.5em}}|@{\hspace{0.5em}}c@{\hspace{0.5em}}|@{\hspace{0.5em}}c@{\hspace{0.5em}}|@{\hspace{0.5em}}c@{\hspace{0.5em}}|@{\hspace{0.5em}}c@{\hspace{0.5em}}||} 
\Xhline{2pt}
$\ell$ & 
$n$    & 
$\Omega$, $\mu = 1.91$ \text{Case A}  & 
$\Omega$, $\mu = 1.91$ \text{Case B}  &
$\Omega$, $\mu = 1.905$ \text{Case A} &
$\Omega$, $\mu = 1.905$ \text{Case B} \\ [0.5ex] 
\Xhline{2pt}\hline
\rule{0pt}{1.15\normalbaselineskip}
$0$ & $0$ & $0.4894-0.1249i$ & $0.8450-0.4863i$ & $0.4280-0.0722i$ & $0.7690-0.3003i$ \\
    & $1$ & $1.3233-0.8988i$ & $1.8444-1.2465i$ & $1.1738-0.5829i$ & $1.6117-0.8343i$ \\ [0.5em]
\hline
\rule{0pt}{1.15\normalbaselineskip}
$1$ & $0$ & $1.7553-0.1381i$ & $1.8743-0.4341i$ & $1.7876-0.0900i$ & $1.8911-0.2857i$ \\
    & $1$ & $2.0976-0.7637i$ & $2.4115-1.1043i$ & $2.0734-0.5052i$ & $2.3205-0.7344i$ \\
    & $2$ & $2.7889-1.4281i$ & $3.1949-1.7263i$ & $2.6173-0.9614i$ & $2.9469-1.1797i$ \\ 
    & $3$ & $3.6057-2.0054i$ & $4.0097-2.2741i$ & $3.2943-1.3886i$ & $3.6493-1.5908i$ \\
    & $4$ & $0.9003-2.4634i$ & $0.9003-2.4633i$ & $4.0060-1.7898i$ & $4.3606-1.9892i$ \\ [0.5em]
\hline
\rule{0pt}{1.15\normalbaselineskip}
$2$ & $0$ & $2.9761-0.1471i$ & $3.0386-0.4477i$ & $3.0520-0.0960i$ & $3.1142-0.2980i$ \\
    & $1$ & $3.1562-0.7586i$ & $3.3275-1.0701i$ & $3.2215-0.5171i$ & $3.3617-0.7435i$ \\
    & $2$ & $3.5563-1.3725i$ & $3.8376-1.6637i$ & $3.5327-0.9631i$ & $3.7404-1.1686i$ \\
    & $3$ & $4.1551-1.9471i$ & $4.4900-2.2266i$ & $3.9866-1.3649i$ & $4.2637-1.5603i$ \\
    & $4$ & $1.6527-2.7014i$ & $1.6527-2.7014i$ & $4.5602-1.7591i$ & $4.8678-1.9634i$ \\ [1ex]
 \Xhline{2pt}
 \end{tabular}
\end{table}

\begin{table}
\centering
\caption{QNMs for electromagnetic perturbations ($s = 1$) of the noncommutative geometry-inspired Schwarzschild wormhole in the nearly extremal case are presented for various values of the mass parameter $\mu$. These numerical values were computed using the spectral method with $200$ polynomials and an accuracy of $200$ digits. "Case A" and "Case B" refer to the boundary conditions introduced in subsections~\ref{subANC} and \ref{subBNC}. Here, $\Omega = 2M\omega$ denotes the dimensionless frequency. Moreover, the mass parameter values $\mu = 1.91$ and $\mu = 1.905$ correspond to throat locations at $x_0 = 0.83098$ and $x_0 = 0.80849$, respectively.} 
\label{table:1aextNCSW}
\vspace*{1em}
\begin{tabular}{||c@{\hspace{0.5em}}|@{\hspace{0.5em}}c@{\hspace{0.5em}}|@{\hspace{0.5em}}c@{\hspace{0.5em}}|@{\hspace{0.5em}}c@{\hspace{0.5em}}|@{\hspace{0.5em}}c@{\hspace{0.5em}}|@{\hspace{0.5em}}c@{\hspace{0.5em}}||}  
\Xhline{2pt}
$\ell$ & 
$n$    & 
$\Omega$, $\mu = 1.91$ \text{Case A}  & 
$\Omega$, $\mu = 1.91$ \text{Case B} &
$\Omega$, $\mu = 1.905$  \text{Case A} &
$\Omega$, $\mu = 1.905$ \text{Case B} \\ [0.5ex]  
\Xhline{2pt}\hline
\rule{0pt}{1.15\normalbaselineskip}
$1$ & $0$ & $1.7213-0.1658i$ & $1.8205-0.5195i$ & $1.7754-0.1093i$ & $1.8758-0.3472i$ \\
    & $1$ & $2.0064-0.9139i$ & $2.2726-1.3300i$ & $2.0496-0.6168i$ & $2.2810-0.9026i$ \\
    & $2$ & $2.6072-1.7375i$ & $2.9840-2.1154i$ & $2.5591-1.1894i$ & $2.8708-1.4669i$ \\
    & $3$ & $3.3755-2.4614i$ & $3.7624-2.7833i$ & $3.2032-1.7313i$ & $3.5457-1.9834i$ \\ [0.5em]
\hline
\rule{0pt}{1.15\normalbaselineskip}
$2$ & $0$ & $2.9606-0.1557i$ & $3.0246-0.4759i$ & $3.0461-0.1012i$ & $3.1108-0.3140i$ \\
    & $1$ & $3.1450-0.8154i$ & $3.3158-1.1720i$ & $3.2255-0.5463i$ & $3.3795-0.7939i$ \\
    & $2$ & $3.5338-1.5371i$ & $3.7934-1.9003i$ & $3.5673-1.0502i$ & $3.7851-1.3093i$ \\
    & $3$ & $4.0837-2.2544i$ & $4.3908-2.5986i$ & $4.0285-1.5668i$ & $4.2930-1.8206i$ \\ [1ex]
 \Xhline{2pt}
 \end{tabular}
\end{table}

\begin{table}
\centering
\caption{QNMs for vector-type gravitational perturbations ($s = 2$) of the noncommutative geometry-inspired Schwarzschild wormhole in the nearly extremal case are presented for various values of the mass parameter $\mu$. These numerical values were computed using the spectral method with $200$ polynomials and an accuracy of $200$ digits. "Case A" and "Case B" refer to the boundary conditions introduced in subsections~\ref{subANC} and \ref{subBNC}. Here, $\Omega = 2M\omega$ denotes the dimensionless frequency.  Moreover, the mass parameter values $\mu = 1.91$ and $\mu = 1.905$ correspond to throat locations at $x_0 = 0.83098$ and $x_0 = 0.80849$, respectively.} 
\label{table:1aaextNCSW}
\vspace*{1em}
\begin{tabular}{||c@{\hspace{0.5em}}|@{\hspace{0.5em}}c@{\hspace{0.5em}}|@{\hspace{0.5em}}c@{\hspace{0.5em}}|@{\hspace{0.5em}}c@{\hspace{0.5em}}|@{\hspace{0.5em}}c@{\hspace{0.5em}}|@{\hspace{0.5em}}c@{\hspace{0.5em}}||}  
\Xhline{2pt}
$\ell$ & 
$n$    & 
$\Omega$, $\mu = 1.91$ \text{Case A}  & 
$\Omega$, $\mu = 1.91$ \text{Case B}  &
$\Omega$, $\mu = 1.905$ \text{Case A} &
$\Omega$, $\mu = 1.905$ \text{Case B} \\ [0.5ex] 
\Xhline{2pt}\hline
\rule{0pt}{1.15\normalbaselineskip}
$2$ & $0$ & $2.8837-0.1922i$ & $2.9695-0.5732i$ & $3.0155-0.1239i$ & $3.0854-0.3760i$ \\
    & $1$ & $3.1469-0.9287i$ & $1.7596-1.1160i$ & $3.2090-0.6267i$ & $3.3777-0.8535i$ \\
    & $2$ & $1.7592-1.1159i$ & $3.4056-1.2388i$ & $3.5967-1.0493i$ & $1.7533-1.1161i$ \\
    & $3$ & $3.7223-1.5163i$ & $4.0717-1.7811i$ & $1.7533-1.1161i$ & $3.8637-1.2307i$ \\
    & $4$ & $4.4343-2.0431i$ & $4.7978-2.3072i$ & $4.1619-1.4125i$ & $4.4764-1.5985i$ \\
    & $5$ & $0.6906-2.3647i$ & $0.6906-2.3647i$ & $4.7993-1.7893i$ & $5.1267-1.9862i$ \\ [0.5em]
\hline
\rule{0pt}{1.15\normalbaselineskip}
$3$ & $0$ & $4.1229-0.1724i$ & $4.1723-0.5210i$ & $4.2749-0.1112i$ & $4.3254-0.3398i$ \\
    & $1$ & $4.2661-0.8739i$ & $4.3974-1.2137i$ & $4.4168-0.5819i$ & $4.5356-0.8374i$ \\
    & $2$ & $4.5702-1.5130i$ & $2.8092-1.6664i$ & $4.6589-1.0974i$ & $4.7719-1.3142i$ \\
    & $3$ & $2.8092-1.6664i$ & $4.8003-1.7730i$ & $4.9354-1.4604i$ & $5.1850-1.6002i$ \\
    & $4$ & $5.0835-2.0252i$ & $5.3933-2.2907i$ & $2.8002-1.6651i$ & $2.8002-1.6651i$ \\
    & $5$ & $5.7094-2.5759i$ & $1.4354-2.7299i$ & $5.4711-1.7755i$ & $5.7675-1.9707i$ \\
    & $6$ & $1.4354-2.7299i$ & $6.0162-2.8977i$ & $6.0739-2.1787i$ & $6.3957-2.3975i$ \\ [1ex]
 \Xhline{2pt}
 \end{tabular}
\end{table}

\section{Conclusions and outlook}

This study presented a comprehensive analysis of QNMs for classical Schwarzschild and noncommutative geometry-inspired wormholes, exploring scalar, electromagnetic, and vector-type gravitational perturbations. By employing a spectral method based on Chebyshev polynomials, we achieved high precision and numerical efficiency, enabling a robust characterization of QNM spectra across various perturbation types.  Our main findings include the absence of overdamped modes and the predominance of oscillatory QNMs, reflecting the linear stability of the investigated wormholes. Moreover, we observed a smooth transition of QNMs from the noncommutative wormhole to the classical Schwarzschild wormhole in the large mass parameter regime, underscoring the accuracy and validity of the computational framework we employed in the present work. Furthermore, our study also revealed significant deviations in the QNM spectra due to noncommutative effects, which manifest in both the real and imaginary parts of the frequencies. 

We have demonstrated that noncommutative effects become most pronounced in the nearly extremal regime, specifically for mass parameter values close to the extremal threshold $\mu_e = 1.904$. In this regime, noncommutative geometry introduces distinct QNM spectra compared to classical Schwarzschild and Morris-Thorne wormholes. Tables XI, XII, and XIII in the manuscript highlight that for values such as $\mu = 1.91$ and $\mu = 1.905$, corresponding to throat locations at $x_0 = 0.83098$ and $x_0 = 0.80849$ respectively, the noncommutative QNM frequencies exhibit significantly different damping rates and oscillatory behaviors.

More precisely, higher damping rates are observed for scalar, electromagnetic, and vector-type gravitational perturbations as one approaches extremality, reflecting the influence of the nearly extreme geometry on the QNM spectra. Conversely, in the limit of large mass parameters ($\mu \gg 1$), noncommutative effects become negligible, and the QNM spectra converge smoothly to those of the classical Schwarzschild wormhole.

However, it is important to note that masses close to the extremal case are likely associated with microscopic wormholes, making their detection through gravitational waves on a cosmic scale improbable. Instead, such objects may be more relevant to future high-energy experiments, such as those conducted in the past at the LHC (Large Hadron Collider), where microscopic noncommutative structures could manifest. This observation places constraints on the direct astrophysical detectability of noncommutative wormholes, while simultaneously highlighting a potential avenue for their exploration in particle physics experiments.

Future research will focus on extending our method to rotating wormholes and higher-dimensional spacetime frameworks. For rotating wormholes, the inclusion of frame-dragging effects in the perturbation equations will present both conceptual and computational challenges, as additional off-diagonal terms in the perturbation operators must be accounted for within the spectral framework. However, the use of Chebyshev polynomials and spectral methods remains a promising approach to achieving accurate and stable numerical results in such scenarios.  The extension to higher-dimensional wormholes will require incorporating additional angular coordinates and modifying the perturbation potentials accordingly. Although this will increase the computational complexity, our spectral method is well-suited for handling such high-dimensional eigenvalue problems, provided adequate computational resources are available. Regarding quantum-gravity corrections, we acknowledge that these effects could indeed be significant, depending on the scale of noncommutativity. As already mentioned in the present work, the noncommutative parameter $\theta$ introduces a minimal length scale that could mitigate singularities, and such smearing effects are already a form of quantum-gravity correction. However, more fundamental quantum-gravity effects, such as those arising from loop quantum gravity or string theory, may introduce corrections that either enhance or suppress the deviations we observed in the QNM spectra. Investigating the interplay between these corrections and noncommutative geometry is an open challenge that we plan to address in future endevours.

\section*{Code availability}

All analytical calculations presented in this document have been verified using the computer algebra system \textsc{Maple}. For transparency and reproducibility, we have included four \textsc{Maple} worksheets that correspond to the analyses conducted in Sections~\ref{WS} and ~\ref{NCW} within the supplementary materials. The discretization of differential operators \eqref{TSCH}, \eqref{TSCHNCW}, and \eqref{TSCHB} using the Chebyshev-type spectral method is equally performed in \textsc{Maple} computer algebra system. Finally, the numerical resolution of the derived quadratic eigenvalue problems, denoted by equation \eqref{eq:eig}, is executed in the \textsc{Matlab} environment utilizing the \texttt{polyeig} function. Access to all these resources is provided through the following \texttt{GitHub} repository link, ensuring that interested parties can freely review and utilize the computational methodologies employed in our study

\begin{itemize}
    \item \url{https://github.com/dutykh/ncwh/}
\end{itemize}


\end{document}